\documentclass[11pt,a4paper,USenglish,thm-restate]{article}

\usepackage{subcaption}

\usepackage{authblk}

\usepackage[dvipsnames]{xcolor}
\usepackage[margin=1in]{geometry}
\usepackage{setspace}
\usepackage{dsfont}
\usepackage{mathrsfs}
\usepackage{amssymb}
\usepackage{amsmath}
\usepackage{amsfonts}
\usepackage{amsthm}
\usepackage{graphicx}
\usepackage{mathtools}
\usepackage{verbatim}
\usepackage{cprotect}
\usepackage{algpseudocode}
\usepackage{pgfplots}
\usepackage[utf8]{inputenc}
\usepackage{listings}
\usepackage{algorithm}
\usepackage{algpseudocode}
\usepackage{xcolor}
\usepackage{hyperref}
\usepackage{enumitem}
\usepackage[capitalize]{cleveref}

\usepackage{graphicx,wrapfig,lipsum}

\usepackage{thm-restate}

\usepackage{cite}

\usetikzlibrary{calc, fit, shapes.geometric}

\newtheorem{definition}{Definition}
\newtheorem{theorem}{Theorem}
\newtheorem{lemma}{Lemma}
\newtheorem{corollary}{Corollary}
\newtheorem{remark}{Remark}

\input{sty/macros.tex}
\newcommand{\parties}{\mathcal{P}} 
\newcommand{\preferences}{\pi}

\newcommand{\sender}{S}

\DeclarePairedDelimiter\abs{\big\lvert}{\big\rvert}

\newcommand{\msgId}{\textsf{id}} 
\newcommand{\msg}{\textsf{m}} 
\newcommand{\timestamp}{\tau}

\newcommand{\ba}{\textsf{BA}} 
 
\newcommand{\bb}{\textsf{BB}}

\newcommand{\king}{\textsf{King}}

\newcommand{\inputt}{\textsf{in}}
\newcommand{\val}{\textsf{value}}
\newcommand{\propose}{\textsf{propose}}

\newcommand{\byzantineSM}{\textsf{bSM}}
\newcommand{\simplifiedSM}{\textsf{sSM}}

\newcommand{\GaleShapley}{\mathcal{A}_\textsf{G-S}}

\title{Byzantine Stable Matching}

\author[1]{Andrei Constantinescu}
\author[1]{Marc Dufay}
\author[2]{Diana Ghinea}
\author[1]{Roger Wattenhofer}

\affil[1]{ETH Z{\"u}rich\\ \{aconstantine, mdufay, wattenhofer\}@ethz.ch}
\affil[2]{Lucerne University of Applied Sciences and Arts\\ diana.ghinea@hslu.ch}

\date{\vspace{-5ex}}

\begin{document}

\maketitle



\vspace{1cm}
\begin{abstract}
\normalsize
In stable matching, one must find a matching between two sets of agents, commonly men and women, or job applicants and job positions. Each agent has a preference ordering over who they want to be matched with. Moreover a matching is said to be stable if no pair of agents prefer each other over their current matching. 
    
We consider solving stable matching in a distributed synchronous setting, where each agent is its own process. Moreover, we assume up to $t_L$ agents on one side and $t_R$ on the other side can be byzantine. After properly defining the stable matching problem in this setting, we study its solvability.

When there are as many agents on each side with fully-ordered preference lists, we give necessary and sufficient conditions for stable matching to be solvable in the synchronous setting. These conditions depend on the communication model used, i.e., if parties on the same side are allowed to communicate directly, and on the presence of a cryptographic setup, i.e., digital signatures.
    
\end{abstract}

\thispagestyle{empty}

\newpage
\tableofcontents

\thispagestyle{empty} 

\newpage
\pagenumbering{arabic}

\section{Introduction}

The \emph{stable matching} problem (also known as \emph{stable marriage}), first introduced by Gale and Shapley \cite{GayleShapley}, has long been a cornerstone of combinatorial optimization and market design. 

It involves finding a stable pairing between two distinct sets of agents --- such as job seekers and job positions or students and universities --- where each participant ranks the opposite set based on individual preferences. The stability criterion dictates that there is no \emph{blocking pair}, i.e., no two unmatched agents should prefer each other over their assigned partners. This foundational task has multiple practical applications in resource assignment and subsequently led to Shapley and Roth winning the Nobel Memorial Prize in Economics in 2012 for their work on \textit{the theory of stable allocations and the practice of market design} \cite{ShapleyRoth}.


In their seminal work \cite{GayleShapley}, Gale and Shapley proved that when the $n$ agents are divided equally into the two sides, and each individual provides a complete preference ranking of the opposite set, a stable matching always exists. Moreover, they provided an algorithm finding such a matching with complexity $O(n^2)$. Further versions of this problem have been considered in \cite{GusfieldIrving}, including variants where the individuals only provide \emph{partial} preferences, or if ties are allowed within the preference rankings. The work of \cite{GusfieldIrving} has shown that a stable matching always exists even in such scenarios, although some individuals may not be matched. 

 The stable matching problem naturally extends to distributed settings, where each agent operates as an independent process or party. Through communication, agents determine their matches while ensuring stability --- a \emph{local} property. 
Furthermore, the Gale-Shapley algorithm inherently functions as a distributed algorithm, as it consists solely of marriage proposals and divorce declaration, both of which can be processed in parallel.



The distributed variant has been studied in various practical scenarios. For instance, the work of Maggs and Sitaraman \cite{MaSi15} explores stable matching in content delivery networks, where stable matching is used for global load balancing by mapping client groups to server clusters. Moreover, stable matching has been leveraged in wireless networks: \cite{BaLoLi11} employs this problem to pair primary and secondary users in a radio network, \cite{ElAhDa12} 
relies on stable matching to pair users and uplink carriers when performing channel assignment, and numerous other studies have explored similar applications \cite{GuZhPa15, BaLoHa12, PaBeSa13}.

We note that, in such scenarios, it is important for the stable-matching-based mechanisms to be resilient to potential faults. Notably, the work of Maggs and Sitaraman \cite{MaSi15} regarding stable matching in content delivery networks has pointed out the potential of (crash) failures. The mitigation strategy proposed by \cite{MaSi15} relies on \emph{leader election}: although crashes do not necessarily degrade the matching obtained, this is a point of failure if the leader misbehaves. To the best of our knowledge, prior works in distributed stable matching assume that parties follow the protocol, or that there is some central trusted unit which can gather all inputs, perform the stable matching algorithm and return the result. This motivates us to investigate scenarios where no such safety exists, and parties may not only crash, but also become byzantine and hence exhibit malicious behavior. Concretely, we ask the following question:

\vspace{-0.1cm}
{
\begin{center}
\emph{Can we achieve stable matching in a network even if some of the parties are byzantine?}
\end{center}
}
\vspace{-0.6cm}

\paragraph{Our Contribution.}
We firstly define the \emph{byzantine stable matching} problem $\byzantineSM$, taking into account that byzantine parties may choose not to participate in the protocol, and preventing honest parties from matching with the same byzantine party. We then investigate the necessary and sufficient conditions for achieving $\byzantineSM$ under various synchronous network topologies, both with and without cryptographic assumptions (digital signatures), and we provide tight conditions. 
We denote the two sets by $L$ and $R$, with $\abs{L} = \abs{R} = k$, and we assume that at most $t_L$ parties in $L$ and at most $t_R$ parties in $R$ may be byzantine.
We consider fully-connected networks, \emph{bipartite} networks (where the parties can only communicate to parties on the other side), and a topology in-between, which we call a \emph{one-sided} networks: this maintains the communication channels of a bipartite network, but additionally provides the parties in side $R$ with complete communication. 
We summarize our findings below:
\begin{itemize}[nosep,leftmargin=*]
    \item When no cryptographic setup is available, $\byzantineSM$ can be solved if and only if: 
    \begin{itemize}
        \item $t_L < k / 3$ or $t_R < k  / 3$ in a fully-connected network.
        \item the following hold in a bipartite network: (i) $t_L, t_R < k/2$; (ii) $t_L < k/3$ or $t_R < k/3$.
        \item the following hold in a one-sided network: (i) $t_R < k/2$; (ii)  $t_L < k/3$ or $t_R < k/3$.
    \end{itemize}
    \item Assuming digital signatures, $\byzantineSM$ can always be solved if the network is fully connected. Otherwise, $\byzantineSM$ can be solved if and only if:
    \begin{itemize}
        \item any of the following holds in a bipartite network: (i) $t_L, t_R < k$; (ii) $t_L < k/3$ or $t_R < k/3$.
        \item $t_R < k$ or $t_L < k / 3$ in a one-sided network.
    \end{itemize}
\end{itemize}

Our settings enable us to establish our conditions' sufficiency by reducing $\byzantineSM$ to Byzantine Broadcast \cite{LSP82}, with one notable exception: this approach falls short in bipartite networks with digital signatures, where one side may be completely byzantine, leaving the honest side disconnected. For this setting, we provide a protocol that simulates a \emph{synchronous fully-connected network with omissions} for the disconnected side, allowing us to achieve $\byzantineSM$.
We also add that our impossibility arguments prove, in fact, even stronger results: even a \emph{simplified} version of $\byzantineSM$ (where parties hold a single favorite as input as opposed to a complete preference list) cannot be solved unless the stated conditions hold.


\paragraph{Related work.}
While the stable matching problem has not been previously explored in the context of byzantine behaviour, malicious strategies such as \emph{lying} were considered. Concretely, Roth \cite{Roth1982TheEO} showed that stable matching is not \emph{truthful}: there are scenarios where an individual can get a more favorable result by lying about their preferences. However, Gale and Shapley \cite{GayleShapley} proved that their algorithm is truthful for one side: an individual on the side doing the proposals can never gain by lying. The case where multiple individuals on the proposing side can collude in the Gale-Shapley algorithm \cite{HuChi06} or where the preference lists can have some mistakes \cite{MaTu18} have also been studied . We note that such adversarial models essentially consider manipulations of the \emph{preferences lists}. In contrast, byzantine fault tolerance protocols are mainly concerned with providing reasonable guarantees even when byzantine parties attempt to prevent honest parties from obtaining \emph{a solution}.


As the Gale-Shapley algorithm \cite{GayleShapley} naturally adapts to distributed settings, lower bounds regarding the number of queries between the two sides have been the topic of interest. Gonczarowski et al \cite{GONCZAROWSKI2019626} provided a lower bound of  $\Omega(n^2)$ boolean queries. However, this does not take into account communication between parties on the same side. In this case, there is still a gap between the best-known lower bound of $\Omega(n^2)$ and upper bound of $O(n^2 \log n)$. When preference lists are similar, Khanchandani and Wattenhofer \cite{Khanchandani2016DistributedSM} describe an algorithm with better complexity and provide a lower bound depending on the similarity of the lists.
Approximation algorithms have also been a topic of interest as a strategy to circumvent the lower bound of \cite{GONCZAROWSKI2019626}. Various definitions for an \emph{approximation} of a stable matching have been analyzed, considering the number of blocking pairs \cite{Ostrovsky}, the number of matches which would have to be broken \cite{GONCZAROWSKI2019626} or how blocking each pair is \cite{Kipnis}.



While some of our necessary conditions enable $\byzantineSM$ to be reduced to well-established problems such as Byzantine Broadcast and Byzantine Agreement \cite{LSP82}, we also encounter settings where $\byzantineSM$ is strictly weaker than these fundamental problems, requiring novel insights.
We also note that the term \emph{matching} has occurred in previous works regarding byzantine faults or \emph{self-stabilization} (starting from an arbitrary state, the system needs to reach a \emph{legitimate configuration} eventually). Most of these works are concerned with finding a \emph{maximal} matching \cite{HSU199277, ChHiSe02, MaMjPi07}, or a maximum matching \cite{HaKa09} in a bipartite graph as opposed to a \emph{stable} matching. A notable exception \cite{LaMa17} focuses on self-stabilizing stable matching. Self-stabilizing protocols assume that all parties will stop being faulty at some point and that no decision is final: any party can decide to \emph{unmatch} at any time. This contrasts with our work which is resilient to some parties being permanently faulty and guarantees that a final decision is reached within a bounded time.



\section{Preliminaries}

The \emph{stable matching} problem can be informally described as the task of pairing two sets of parties in such a way that 
no two unmatched parties should prefer each other over their assigned partners.
Stable matching problems have been typically framed in contexts such as matching men with women, students with universities, or producers with consumers.

\paragraph{Standard stable matching.}
We consider a set of $n = 2k$ parties $\parties$, which are divided into two disjoint sets $L$ and $R$ with $|L| = |R| = k$: $L$ can represent, for instance, the set of men/students/producers while $R$ can represent the set of women/universities/consumers. 
In the \emph{stable matching problem}, every party $u$ in $L$ (resp.~$R$) has as input a \emph{preference list} (a permutation) $\preferences_u$ over the parties in $R$ (resp.~$L$).  We say that $u$ \emph{prefers} $v$ over $w$ if $v$ appears before $w$ in $\preferences_u$. In addition, $u$ prefers any party in its preference list $\preferences_{u}$ over being alone.


The objective of this problem
is to determine a \emph{stable matching}: a matching $M$ between $L$ and $R$ such that there is no \emph{blocking} pair.
A (non-matched) pair of parties $(u, v) \in L \times R$ is \emph{blocking} if $u$ and $v$ prefer each other compared to who they are currently matched to. 
As parties always prefer being matched to being alone, a pair of two unmatched parties on opposite sides is considered blocking.
This implicitly requires the matching to be \emph{maximal} (all parties are matched). 
It has been proven that
a stable matching always exists, and it can be found using the Gale-Shapley algorithm $\GaleShapley$ \cite{GayleShapley}.
\begin{theorem}[\hspace{-1pt}{\cite{GayleShapley}}] \label{theorem:gale-shapley}
    There is a deterministic algorithm $\GaleShapley$ that takes as input the preference lists $\pi$ of all parties in $L$ and $R$ and returns a stable matching $M$.
\end{theorem}

\paragraph{Stable matching in a network.}
We now define the stable matching problem in a distributed setting. Here, the parties in $L \cup R$ are processors running a protocol over a network, exchanging messages via \emph{bidirectional authenticated communication channels}.
We assume that the network is \emph{synchronous}: the parties have synchronized clocks, all parties start at time $0$, and every message is delivered within a publicly known amount of time $\Delta$. This 
allows protocols to operate in rounds. As depicted in \cref{fig:communication-networks}, we will explore different network topologies, described below.

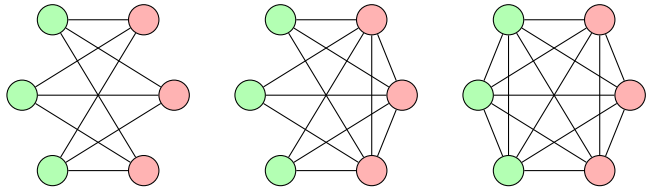
\begin{figure}[h]
    \centering
\begin{tikzpicture}[node distance=1.5cm, every node/.style={circle, draw,fill=red!30, minimum size=4mm}]
    \begin{scope}[local bounding box=graph1]
        \foreach \y in {0,2} {
            \node[fill=green!30] (L\y1) at (0.4,\y) {};
        }
        \node[fill=green!30] (L11) at (0,1) {};
        \foreach \y in {0,2} {
            \node (R\y1) at (1.6,\y) {};
        }
        \node (R11) at (2,1) {};
        \foreach \a in {0,1,2} {
            \foreach \b in {0,1,2} {
                \draw (L\a1) -- (R\b1);
            }
        }
    \end{scope}

    \begin{scope}[shift={(3,0)}, local bounding box=graph2]
        \foreach \y in {0,2} {
            \node[fill=green!30] (L\y2) at (0.4,\y) {};
        }
        \node[fill=green!30] (L12) at (0,1) {};
        \foreach \y in {0,2} {
            \node (R\y2) at (1.6,\y) {};
        }
        \node (R12) at (2,1) {};
        \foreach \a in {0,1,2} {
            \foreach \b in {0,1,2} {
                \draw (L\a2) -- (R\b2);
            }
        }
        \draw (R02) -- (R12) -- (R22) -- (R02);
    \end{scope}

    \begin{scope}[shift={(6,0)}, local bounding box=graph3]
        \foreach \y in {0,2} {
            \node[fill=green!30] (L\y3) at (0.4,\y) {};
        }
        \node[fill=green!30] (L13) at (0,1) {};
        \foreach \y in {0,2} {
            \node (R\y3) at (1.6,\y) {};
        }
        \node (R13) at (2,1) {};
        \foreach \a in {0,1,2} {
            \foreach \b in {0,1,2} {
                \draw (L\a3) -- (R\b3);
            }
        }
        \draw (L03) -- (L13) -- (L23) -- (L03);
        \draw (R03) -- (R13) -- (R23) -- (R03);
    \end{scope}
\end{tikzpicture}
\setlength{\belowcaptionskip}{-10pt}
\caption{The different kinds of communication networks we consider. From left to right: bipartite, one-sided and fully-connected networks. Note that even when communication is possible within $L$ or $R$, the matching is still between parties on opposite sides (not within $L$ or $R$).} 
\label{fig:communication-networks}
\end{figure}

\noindent \emph{(Fully-connected network)} Parties are pairwise connected. This model is relevant in scenarios such as forming partnerships within a close-knit social group.

\noindent \emph{(One-sided network)} Parties are pairwise connected, except parties within $L$, which cannot communicate directly. This structure is applicable in contexts such as kidney donations, where privacy constraints prevent recipients from directly interacting with each other.

\noindent \emph{(Bipartite network)}
Only pairs of parties in $L \times R$ are connected.
This setup is relevant in cases such as matching international job applicants, where communication is restricted solely to potential matches across the two sets.

\vspace{0.1cm}

We remark that each model is strictly stronger than the previous one. 




The parties will be then running a protocol $\Pi$ where each party holds a preference list as input, and each party obtains as output its match (from the opposite side). In this setting, $\Pi$ achieves \emph{distributed stable matching} if the following properties hold:
\emph{(Termination)} Each party outputs a party on the opposite side to match with;
\emph{(Symmetry)} If party $u$ decides to match party $v$, then $v$ decides to match party $u$;
\emph{(Stability)} There are no blocking pairs.


\paragraph{Faults.}
So far, we have defined the stable matching problem in a \emph{fault-free} setting. From now on, we assume an adversary that may (permanently) corrupt up to $t_L$ parties in $L$ and up to $t_R$ parties in $R$. The corrupted parties become \emph{byzantine}: they may deviate arbitrarily (even maliciously) from the protocol. A party is \emph{honest} if it never became byzantine.
Our protocols will assume that the adversary is \emph{adaptive}: it may choose
to corrupt parties at any point of the protocol's execution. Our impossibility results, however, hold even against a \emph{static} adversary, which needs to choose which parties to corrupt at the beginning of the protocol's execution.


\paragraph{Refining the Problem.} 
Byzantine parties require us to refine the definition of the stable matching problem. First, we need to take into account that our properties should only be concerned with the outputs of \emph{honest} parties. Second, the byzantine parties may choose not to participate in the protocol, preventing us from obtaining a maximal matching. Consequently, we adjust the previous properties as follows:

\vspace{0.1cm}
\noindent \emph{(Termination)} Every \emph{honest} party outputs: either a party on the opposite side \emph{or nobody}.

\noindent \emph{(Symmetry)} For 
two \emph{honest} parties $u$ and $v$, if $u$ decides to match $v$, then $v$ decides to match $u$.

\noindent \emph{(Stability)} There are no blocking pairs made of \emph{honest} parties.

\vspace{0.1cm}



\begin{wrapfigure}{r}{3cm}
\centering
\includegraphics[width=3cm]{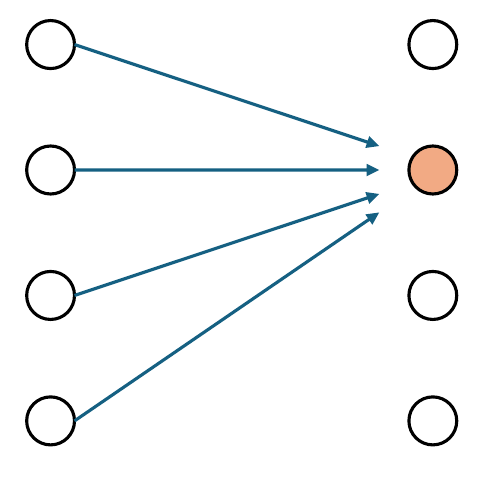}
\vspace{-1.5cm}
\end{wrapfigure} 

Note that these properties are not strong enough to lead to a \emph{relevant} matching: multiple honest parties may be matched to the same byzantine party (in the figure on the right, if the orange party is byzantine, the depicted matching satisfies 
symmetry and stability).
%
Therefore, we introduce an additional intuitive condition that prevents such scenarios:

\vspace{0.1cm}
\noindent \emph{(Non-competition)} If two honest parties output $u, v \in L \cup R$, then $u \neq v$.
\vspace{0.15cm}


We may now present the formal definition of byzantine stable matching. 
\begin{definition}[Byzantine Stable Matching ($\byzantineSM$)]
Consider a protocol $\Pi$ where every party in $L \cup R$ holds as input a preference list over the parties on the other side. Then, $\Pi$ achieves byzantine stable matching ($\byzantineSM$) with respect to $t_L$ and $t_R$ if it satisfies the following even when up to $t_L$ parties in $L$ and $t_R$ parties in $R$ are byzantine: termination, symmetry, stability, non-competition.
\end{definition}


\paragraph{Cryptographic Assumptions.}
As we will see, the solvability of $\byzantineSM$ in a given setting depends on whether we assume a trusted setup and cryptographic primitives. We will use the term \emph{unauthenticated setting} to refer to settings where no cryptographic assumptions are made. In contrast, we use the term \emph{authenticated setting} to refer to a setting where a public key infrastructure and a secure digital signature scheme are available. For simplicity of presentation, we assume that signatures are unforgeable. When replaced with real-world instantiations, our feasibility results in the authenticated setting still hold except for negligible probability (in the scheme's security parameter) against computationally-bounded adversaries.



\paragraph{Warm-up solution.} 
Synchronous networks come with communication primitives that often us to reduce $\byzantineSM$ to an offline problem. One such primitive is \emph{Byzantine Broadcast} ($\bb$) \cite{LSP82}.
\begin{definition}[Byzantine Broadcast ($\bb$)]\label{def:bc}
	Let $\Pi$ be a protocol where a designated party $\sender$ (the sender) holds a value $v_{\sender}$. 
    We say that $\Pi$ achieves Byzantine Broadcast ($\bb$) if the following hold even when up to some number $t$ of the parties are corrupted (alternatively for our setting, up to $t_L$ in $L$ and up to $t_R$ in $R$):
    \begin{itemize}[nosep] 
    \item \emph{(Termination)} All honest parties output;
    \item \emph{(Validity)} If $\sender$ is honest, every honest party outputs $v_{\sender}$;
    \item \emph{(Consistency)} Honest parties output the same value.
    \end{itemize}
    
    
    
    
\end{definition}


A $\bb$ protocol allows the sender to disseminate its preferences so that all parties obtain identical views of them. If each party runs an invocation of $\bb$ to distribute its preferences, then by the end, all parties will have identical views of everyone's preferences. This enables them to run $\GaleShapley$ offline and obtain the same stable matching, thereby solving $\byzantineSM$. 
This provides us with the lemma below. We include the formal proof in Appendix~\ref{appendix:preliminaries}.


\begin{restatable}{lemma}{BroadcastEasy} \label{lemma:broadcast-easy}
Whenever $\bb$ is available, $\byzantineSM$ is solvable.
\end{restatable}

\section{Simplified Stable Matching}
For most of our impossibility results, we do not require the inputs to be complete preference lists: we rely only on parties' favorites. Therefore, we introduce the \emph{simplified stable matching} problem ($\simplifiedSM$), which mostly follows the same rules as $\byzantineSM$. The main difference is that a party's input is a party on the other side, not a preference list. If party $u$ has as input party $v$, we say that $u$'s \emph{favorite} is $v$. The stability property is then replaced by simplified stability:

\vspace{.1cm}
\noindent \emph{(Simplified stability)} If two honest parties are each other's favorites, they output each other.
\vspace{.1cm}

Our impossibility proofs will describe settings where a protocol cannot simultaneously achieve termination, symmetry, non-competition, and this simplified property. In Appendix \ref{appendix:simplified-stable-matching}, we show that $\simplifiedSM$ can be reduced to $\byzantineSM$, enabling us to state the following result.



\begin{restatable}{lemma}{SimplifiedReduction} \label{coro:to-simplified}
Whenever $\simplifiedSM$ is not solvable, $\byzantineSM$ is not solvable.
\end{restatable}

We finish with a helpful technical lemma
allowing our impossibility arguments to only focus on proving that $\simplifiedSM$ cannot be solved in settings with few parties. The lemma then generalizes our arguments to settings with more parties. The proof is enclosed in Appendix \ref{appendix:simplified-stable-matching}
\begin{restatable}{lemma}{ReduceNumberLemma}\label{lemma:reduce-number}
Let $\Pi$ be a protocol solving $\simplifiedSM$, supporting up $t_L$ byzantine parties in $L$ and $t_R$ byzantine parties in $R$. Then, for any $0 < d \leq k = n/2$, there exists a protocol $\Pi'$ solving $\simplifiedSM$ on $2d$ parties ($d$ on each side) that supports up to $\lfloor \frac{t_L}{\lceil k / d \rceil} \rfloor$ byzantine parties on the left side and $\lfloor \frac{t_R}{\lceil k / d \rceil} \rfloor$ byzantine parties on the right side.
\end{restatable}

\section{Solvability in Unauthenticated Settings} \label{section:no-pki}
In this section, we describe tight conditions for solving $\byzantineSM$ in unauthenticated settings (no cryptographic assumptions). In the following, we first present our findings in the fully-connected network case. Afterwards, we focus on the one-sided and bipartite network cases.

\subsection{Fully-Connected Network}

The conditions for the fully-connected network case are presented in \cref{theo:pki-complete}, stated below.
\begin{theorem} \label{theo:pki-complete}
$\byzantineSM$ is solvable in a fully-connected unauthenticated network if and only if $t_L < k/3$ or $t_R < k/3$.
\end{theorem}

We first focus on the feasibility part, for which we recall \cref{lemma:broadcast-easy}: if $\bb$ can be achieved in a setting, then $\byzantineSM$ also can. As stated in the lemma below, $\bb$ is, in fact, solvable for our conditions.
Hence, the two lemmas together enable us to conclude that $\byzantineSM$ is solvable in a fully-connected unauthenticated network whenever  $t_L < k / 3$ or $t_R < k / 3$, as desired. 
\cref{lemma:general-adversaries} is a corollary of \cite[{Theorem 2}]{DISC:FitMau98}. We note that \cite{DISC:FitMau98} focuses on \emph{general adversaries}.
Roughly, this is an adversarial model where the adversary has to choose which parties to corrupt from a predefined subset-closed list of options. We provide a detailed discussion, along with the proof of \cref{lemma:general-adversaries}, in Appendix \ref{appendix:general-adversaries}.

\begin{restatable}{lemma}{GeneralAdversaries} \label{lemma:general-adversaries}
    $\bb$ is solvable in a fully-connected network if $t_L < k/3$ or $t_R < k/3$.
\end{restatable}




We now show that at least one of the conditions $t_L < k / 3$ and $t_R < k / 3$ holding is necessary. To do so, we prove this property for the special case $n=6$ for $\simplifiedSM$. One of our proof's ingredients is the shifting scenarios proof technique from \cite{PODC:FisLynMer85}, i.e., defining a \emph{larger} system to reach a contradiction. 
Then, \cref{lemma:reduce-number} enables us to conclude that for arbitrary $n$, $\simplifiedSM$ is not solvable if  $t_L \geq k/3$ and $t_R \geq k/3$. Finally, Lemma~\ref{coro:to-simplified} lifts this impossibility result to $\byzantineSM$, completing the proof of \cref{theo:pki-complete}.
\begin{lemma}\label{lemma:pki-6}
Assume a fully-connected unauthenticated network and $n = 6$. Then, no protocol achieves $\simplifiedSM$ for $t_L = t_R = 1$.
\end{lemma}
\begin{proof}
Assume for a contradiction that $\Pi$ achieves $\simplifiedSM$ in this setting.
We start with a high-level outline of the proof. We will construct a larger system (i.e., for 12 parties) by `duplicating' the communication graph and consider running $\Pi$ in this new system, with each party following the intended behavior of the corresponding party in the original system. By choosing specific pairs of byzantine parties and their strategies in the original system, the adversary will be able to simulate being in the new system towards the honest parties. Going even further, we will present three different setups of the original system where the adversary can achieve this simulation. Each setup will provide different insights into how the parties should behave in the larger system, leveraging that $\Pi$ is correct for any setup of the original system. However, these findings will ultimately be contradictory, proving that $\Pi$ cannot exist.


\begin{figure}[h]
    \centering
    \begin{tikzpicture}[every node/.style={scale=0.65}]

\def\innerRadius{1.5}
\def\outerRadius{3}

\begin{scope}[local bounding box=graph1]
\foreach \i/\name/\pos in {0/a_1/right, 60/b_1/above right, 120/c_1/above left, 
                           180/a_2/left, 240/b_2/below left, 300/c_2/below right} {
    \node[draw, circle, fill=green!30, minimum size=4mm] (\name) at (\i:\innerRadius) {$\name$};
}

\foreach \i/\name/\pos in {0/u_1/right, 60/v_1/above right, 120/w_1/above left, 
                           180/u_2/left, 240/v_2/below left, 300/w_2/below right} {
    \node[draw, circle, fill=red!30, minimum size=4mm] (\name) at (\i:\outerRadius) {$\name$};
}

\foreach \a/\b in {a_1/b_1, b_1/c_1, c_1/a_2, a_2/b_2, b_2/c_2, c_2/a_1} {
    \draw (\a) -- (\b);
}

\foreach \a/\b in {u_1/v_1, v_1/w_1, w_1/u_2, u_2/v_2, v_2/w_2, w_2/u_1} {
    \draw (\a) -- (\b);
}

\foreach \a\b in {u_1/a_1, u_1/b_1, u_1/c_2, 
    u_2/a_2, u_2/b_2, u_2/c_1,
    v_1/b_1, v_1/c_1, v_1/a_1,
    v_2/b_2, v_2/c_2, v_2/a_2,
    w_1/c_1, w_1/a_2, w_1/b_1,
    w_2/c_2, w_2/a_1, w_2/b_2} {
    \draw[dotted] (\a) -- (\b);
}

\foreach \a\b in {c_1/v_1, v_1/c_1, a_2/v_2, v_2/a_2} {
    \draw[->, OliveGreen, very thick] (\a) -- ($(\a)!0.3!(\b)$);
}
\end{scope}

\begin{scope}[shift={(7,-7)},local bounding box=graph2]

\foreach \i/\name/\tag/\pos in {330/t/a/below right,  90/c_1/b/above, 
                           210/a_2/c/below left} {
    \node[draw, circle, fill=green!30, minimum size=4mm] (\tag) at (\i:\innerRadius) {$\name$};
}

\foreach \i/\name/\tag/\pos in {330/t/u/below right,  90/w_1/v/above, 
                           210/u_2/w/below left} {
    \node[draw, circle, fill=red!30, minimum size=4mm] (\tag) at (\i:\outerRadius) {$\name$};
}

\foreach \a/\b in {a/b, b/c, c/a} {
    \draw (\a) -- (\b);
}

\foreach \a/\b in {u/v, v/w, w/u} {
    \draw (\a) -- (\b);
}

\foreach \a\b in {u/a, u/b, u/c,
    v/a, v/b, v/c,
    w/a, w/b, w/c} {
    \draw[dotted] (\a) -- (\b);
}

\foreach \a\b in {b/u, c/u} {
    \draw[->, OliveGreen, very thick] (\a) -- ($(\a)!0.3!(\b)$);
}

\node[ellipse, draw, orange, fill=orange!60, thick, fit=(a) (u), rotate=330, inner xsep=10pt, inner ysep=0.5pt, dashed] {};

\end{scope}

\begin{scope}[shift={(0,-7)},local bounding box=graph3]

\foreach \i/\name/\tag/\pos in {330/b_1/a/below right,  90/c_1/b/above, 
                           210/a_2/c/below left} {
    \node[draw, circle, fill=green!30, minimum size=4mm] (\tag) at (\i:\innerRadius) {$\name$};
}

\foreach \i/\name/\tag/\pos in {330/v_1/u/below right,  90/w_1/v/above, 
                           210/u_2/w/below left} {
    \node[draw, circle, fill=red!30, minimum size=4mm] (\tag) at (\i:\outerRadius) {$\name$};
}

\foreach \a/\b in {a/b, b/c, c/a} {
    \draw (\a) -- (\b);
}

\foreach \a/\b in {u/v, v/w, w/u} {
    \draw (\a) -- (\b);
}

\foreach \a\b in {u/a, u/b, u/c,
    v/a, v/b, v/c,
    w/a, w/b, w/c} {
    \draw[dotted] (\a) -- (\b);
}

\foreach \a\b in {b/u, u/b} {
    \draw[->, OliveGreen, very thick] (\a) -- ($(\a)!0.3!(\b)$);
}

\node[ellipse, draw, orange, fill=orange!60, thick, fit=(c) (w), rotate=210, inner xsep=10pt, inner ysep=0.5pt, dashed] {};

\end{scope}

\begin{scope}[shift={(7,-1)},local bounding box=graph4]

\foreach \i/\name/\tag/\pos in {330/b_2/a/below right,  90/c_1/b/above, 
                           210/a_2/c/below left} {
    \node[draw, circle, fill=green!30, minimum size=4mm] (\tag) at (\i:\innerRadius) {$\name$};
}

\foreach \i/\name/\tag/\pos in {330/v_2/u/below right,  90/w_1/v/above, 
                           210/u_2/w/below left} {
    \node[draw, circle, fill=red!30, minimum size=4mm] (\tag) at (\i:\outerRadius) {$\name$};
}

\foreach \a/\b in {a/b, b/c, c/a} {
    \draw (\a) -- (\b);
}

\foreach \a/\b in {u/v, v/w, w/u} {
    \draw (\a) -- (\b);
}

\foreach \a\b in {u/a, u/b, u/c,
    v/a, v/b, v/c,
    w/a, w/b, w/c} {
    \draw[dotted] (\a) -- (\b);
}

\foreach \a\b in {c/u, u/c} {
    \draw[->, OliveGreen, very thick] (\a) -- ($(\a)!0.3!(\b)$);
}

\node[ellipse, draw, orange, fill=orange!60, thick, fit=(b) (v), rotate=0, inner xsep=10pt, inner ysep=8pt, dashed] {};

\end{scope}
    
\end{tikzpicture}
    \vspace{-1cm}
    \setlength{\belowcaptionskip}{-10pt}
    \caption{(i) Top left: system constructed in the proof of \cref{lemma:pki-6}; (ii) Top right: indistinguishable execution where $a_2, b_2, u_2$ and $v_2$ are correct while the remaining nodes are simulated by the byzantine parties; (iii) Bottom left: indistinguishable execution where $b_1, c_1, v_1$ and $w_1$ are correct; (iv) Bottom right: indistinguishable execution where $c_1, a_2, u_2$ and $w_1$ are correct.}
    \label{fig:no-pki-fully-model}
\end{figure}
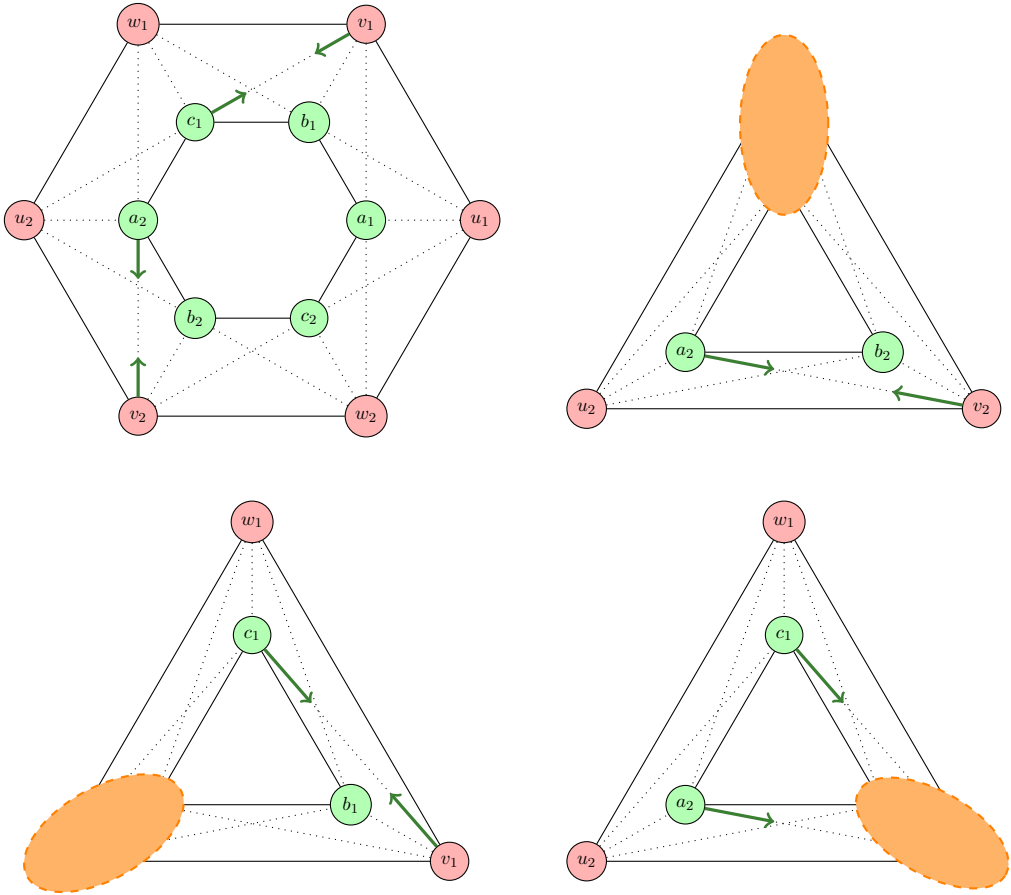

We denote the three parties in $L$ by $a, b, c$, and the three parties in $R$ by $u, v, w$.
By duplicating each party, we obtain a graph with $12$ nodes such that each node is connected to $3$ parties on the opposite side, as depicted in \cref{fig:no-pki-fully-model} (i). 
Consider an execution of $\Pi$ in the new system with the following inputs: $c_1$ and $v_1$ have each other as their favorite, $a_2$ and $v_2$ have each other as their favorite, and all other inputs are arbitrary.
So far, because $\Pi$ is running in a non-standard network, we cannot say anything about the output of the parties (or even whether they terminate or not). 
However, using indistinguishability arguments,
we will prove that there exists a normal-operation scenario for $\Pi$ where two honest parties match the same party, giving a contradiction:

\begin{itemize}[nosep,leftmargin=*]
    
    \item First, we consider the setting where $a_2, b_2, u_2$ and $v_2$ are honest while all the remaining nodes are being simulated (internally) by $c$ and $w$, which are byzantine. This matches \cref{fig:no-pki-fully-model} (ii). This setting is valid for $\Pi$: there is at most one byzantine party in $L$ and one byzantine party in $R$. As such, $a_2, b_2, u_2$ and $v_2$ must terminate. Moreover using the simplified stability property, because $a_2$ and $v_2$ prefer each other, we get that $a_2$ decides to match party $v_2$. 

    \item Then, we consider the setting where $b_1, c_1, v_1$ and $w_1$ are honest, while the remaining nodes are being simulated by two byzantine parties $a$ and $u$. This matches \cref{fig:no-pki-fully-model} (iii). Similarly to before, we get that $c_1$ must terminate and decide to match party $v_2$.

    \item The last setting is where $c_1, a_2, u_2$ and $w_1$ are honest, while the remaining nodes are being simulated by two byzantine parties $b$ and $v$. This matches \cref{fig:no-pki-fully-model} (iv). We remark that $a_2$ (resp.~$c_1$) cannot distinguish between this setting and \cref{fig:no-pki-fully-model} (ii) (resp.~(iii)). As such, as we proved above, $a_2$ and $c_1$ will both decide to match $v$ (previously we had written $v_1$ and $v_2$ to distinguish between the two copies, but here there is only one). Moreover, this setting is valid for $\Pi$, which means that its output should satisfy the properties of stable matching. However, $a_2$ and $c_1$ are both honest and decide to match the same party $v$, breaking the non-competition rule, hence we obtain a contradiction. \qedhere
\end{itemize}

\end{proof}

\subsection{Bipartite and One-Sided Networks}

The theorems below give the necessary and sufficient conditions for the bipartite and one-sided communication cases. Note that, in contrast to the fully-connected case, each theorem requires one additional condition (i) to hold on top of the previous condition (ii) that was already required for the fully-connected case.

\begin{theorem}\label{theorem:main:no-pki-bipartite}
$\byzantineSM$ is solvable in a bipartite unauthenticated network if and only if both of these conditions are satisfied: (i) $t_L, t_R < k/2$; (ii) $t_L < k/3$ or $t_R < k/3$.
\end{theorem}

\begin{theorem}\label{theorem:main:no-pki-one-sided} $\byzantineSM$ is solvable in a one-sided unauthenticated network if and only if both of these conditions are satisfied: (i) $t_R < k/2$; (ii)  $t_L < k/3$ or $t_R < k/3$.
\end{theorem}




The main idea we will use to get the feasibility parts of the previous theorems 
is that for the stated conditions, we may actually return to assuming
a fully-connected network.
In particular, two parties $u, v$ on the same side that do not share a communication channel may simulate such a channel between them by using the parties on the opposite side as a proxy: $u$ sends the desired message to all parties on the other side, who in turn forward it to $v$, which takes a majority vote to decide on the sent message. This works as long as there is an honest majority on the other side, giving us the following lemma (formal proof in 
Appendix~\ref{appendix:complete-bipartite-graph-without-pki}).
%
\begin{restatable}{lemma}{BipartiteCommunicationNoPKI}\label{lemma:pki-bipartite}
Let $S$ and $S'$ denote the two sides.
If the parties in $S$ are disconnected and $t_{S'} < k /2$, we may assume the parties in $S$ are fully-connected.
\end{restatable}


This provides us with the corollaries below, enabling us to prove the stated conditions to be sufficient when combined with \cref{theo:pki-complete}, which has assumed a fully-connected network.
\begin{corollary}\label{lemma:pki-one-sided}
In a one-sided network, we may assume a fully-connected network if $t_L < k / 2$.
\end{corollary}
\begin{corollary}\label{lemma:pki-complete}
In a bipartite network, we may assume a fully-connected network if $t_L, t_R < k/2$.
\end{corollary}

To prove that the conditions in \cref{theorem:main:no-pki-bipartite,theorem:main:no-pki-one-sided} are also necessary, we make use of the lemma below. Similarly to the proof of \cref{lemma:pki-6}, our argument includes a technique from \cite{PODC:FisLynMer85} that enables us to reach a contradiction by working with a non-standard system.
\begin{lemma} \label{lemma:pki-4}
Assume a one-sided unauthenticated network and $n = 4$. Then, no protocol achieves $\simplifiedSM$ for $t_L = 0$ and $t_R = 1$.
\end{lemma}

\begin{proof}
Write $a, b$ for the nodes in $L$ and $c, d$ for the nodes in $R$. All nodes are connected in the communication network except $a$ and $b$. We will prove that no such protocol exists even in the no-harder setting where \emph{exactly} one party in $R$ is corrupted, which we henceforth assume.

It will suffice to prove the impossibility for the bipartite network case, i.e., without the edge $c$-$d$.
In particular, we claim that messages sent across this edge cannot be helpful. To see this intuitively, recall our assumption that \emph{exactly} one party in $R$ is byzantine, say $d$. Party $c$ knows that $d$ is byzantine, meaning any messages received from $d$ could be completely arbitrary. Therefore, $c$ may as well simulate receiving them by replacing them with a default value. Henceforth, we assume that the communication network is bipartite.

Assume for a contradiction that $\Pi$ is a protocol achieving $\simplifiedSM$ for $n = 4$ parties $L = \{a, b\}$ and $R = \{c, d\}$ in a bipartite unauthenticated network given that no party in $L$ is corrupted and \emph{exactly} one party in $R$ is corrupted. 

The key insight in our proof is that the bipartite communication network actually forms the undirected cycle $a$-$c$-$b$-$d$-$a$. We construct a larger system by duplicating each party and linking them into a cycle twice as long: $a_1$-$c_1$-$b_1$-$d_1$-$a_2$-$c_2$-$b_2$-$d_2$-$a_1$. This is depicted in the first row of \cref{fig:cycle}.
We consider running $\Pi$ in this setting by running the protocol used for $a$ on $a_1$ and $a_2$, the protocol used for $b$ on $b_1$ and $b_2$, and so on. We will now assign favorites (inputs) to the vertices: we make $a_1$ and $c_1$ each other's favorites and $b_2$ and $c_2$ each other's favorites. Other vertices are assigned favorites arbitrarily. We will show that by running protocol $\Pi$ in this setting, we get a contradiction.

\begin{figure}[t]
\centering
\includegraphics[scale=0.49]{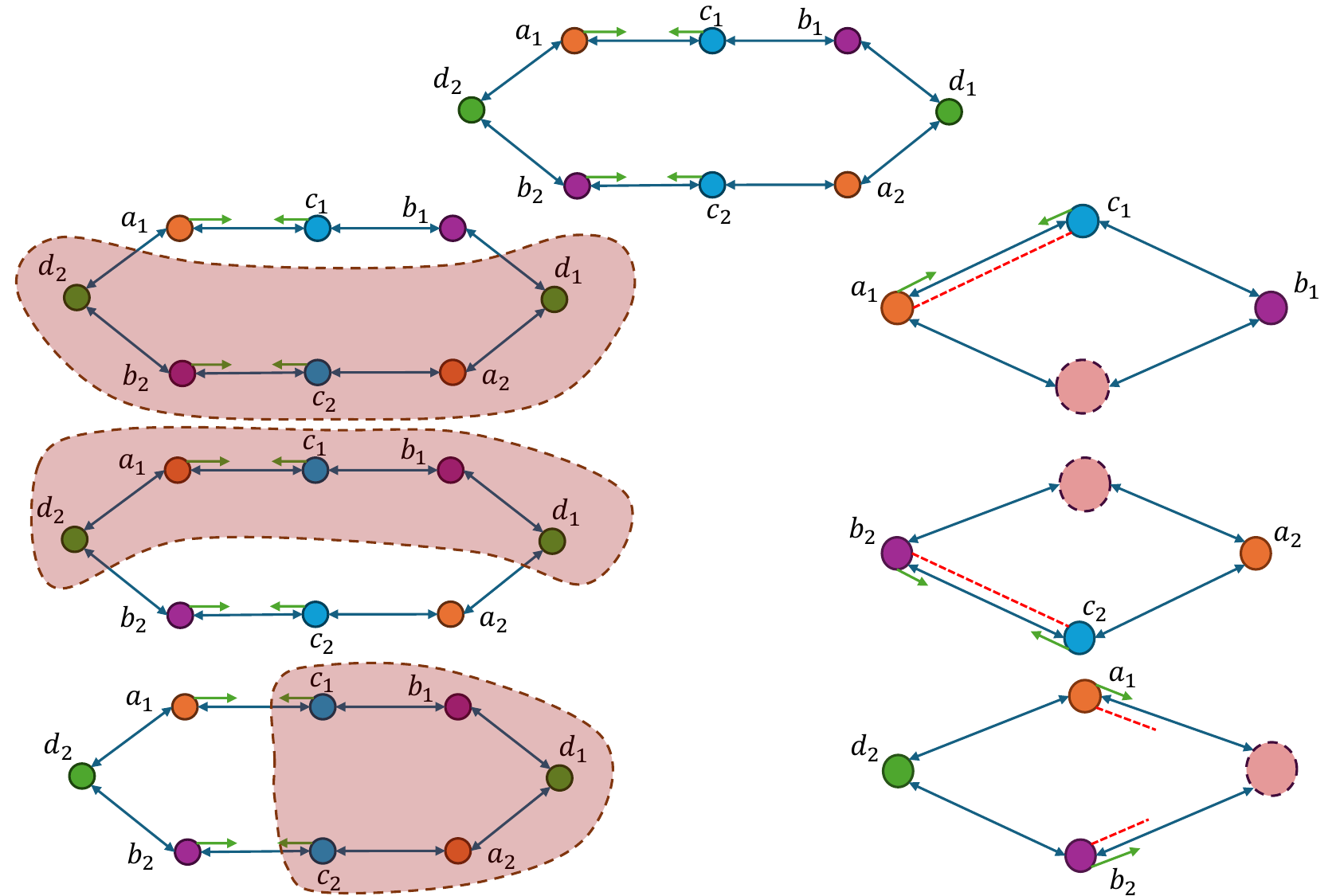}
\setlength{\belowcaptionskip}{-10pt}
\caption{The eight parties in a cycle.
In the first case, $a_1$ and $c_1$ are both honest and each other's favorites. Therefore, they must match each other. In the second case, we similarly get that $b_2$ and $c_2$ must match. Finally, in the last case, we get that both $a_1$ and $b_2$ try to match party $c$ which is byzantine, but the non-competition property does not allow this.}
\label{fig:cycle}
\end{figure}

First (second row in \cref{fig:cycle}), we consider the case where $a_1$, $c_1$ and $b_1$ are honest while $d$ is byzantine and simulating $d_1$-$a_2$-$\dots$-$d_2$. In this case, because $a_1$ and $c_1$ are honest and both each other's favorites, by simplified stability, they must match each other.

By symmetry (third row in \cref{fig:cycle}), $b_2$ and $c_2$ must match each other.

Last (forth row in \cref{fig:cycle}), we consider the case where $b_2$, $d_2$ and $a_1$ are honest while $c$ is byzantine and simulating $c_1$-$b_1$-$\dots$-$c_2$. In this case, $a_1$ and $b_2$ are both honest and both decide to match the same party $c$ (by the previous two cases), which is prohibited by non-competition, giving us a contradiction.
\end{proof}

We conclude the section by presenting the proofs of \cref{theorem:main:no-pki-bipartite}, giving the conditions for bipartite networks, and \cref{theorem:main:no-pki-one-sided}, providing the conditions for one-sided networks.
\begin{proof}[Proof of \cref{theorem:main:no-pki-bipartite}]
\cref{lemma:pki-bipartite} implies that the bipartite communication model is weaker than the one-sided communication model, and therefore \cref{theorem:main:no-pki-one-sided} enables us to conclude that the conditions in our Theorem's statement are necessary. Note that, for the condition $t_L < k / 2$ we need to apply \cref{lemma:pki-bipartite} for $S = L$ and $S' = R$, while for $t_R < k / 3$ we need $S = R$ and $S' = L$. For sufficiency, the condition $t_L, t_R < k/2$ enables us to apply \cref{lemma:pki-complete} and assume a fully-connected network. Afterwards, \cref{theo:pki-complete} ensures that $\byzantineSM$ is solvable.
\end{proof}

\begin{proof}[Proof of \cref{theorem:main:no-pki-one-sided}]
\cref{lemma:pki-4} proves that for $n = 4$ parties, $\simplifiedSM$ cannot be achieved in this setting if $t_R \geq 1$.
Then, using \cref{lemma:reduce-number}, this proves that for an arbitrary $n$, $\simplifiedSM$ cannot be achieved if $t_R \geq k/2$ in a one-sided network. Therefore, using \cref{coro:to-simplified}, this proves that $t_R < k/2$ is necessary for $\byzantineSM$ as well.
When $t_R < k / 2$ holds, \cref{lemma:pki-complete} enables us to assume a fully-connected network. Then, we may conclude using \cref{theo:pki-complete} that it is both necessary and sufficient that one of the conditions $t_L < k/3$, $t_R < k/2$ holds.
\end{proof}

\section{Solvability in Authenticated Settings}

We now consider authenticated settings. Assuming digital signatures will enable us to solve $\byzantineSM$ up to much higher corruption thresholds in contrast to the unauthenticated case. This section is organized similarly to Section \ref{section:no-pki}: we first analyze the fully-connected network case, and afterwards focus on one-sided and bipartite networks.

\subsection{Fully-Connected Network} 
In the fully-connected network case, $\byzantineSM$ is always solvable: we utilize the Dolev-Strong protocol \cite{DolStr83}, which achieves $\bb$ resilient against $t < n$ corruptions assuming PKI. Then, \cref{lemma:broadcast-easy} directly implies the theorem below.
\begin{theorem}\label{theo:with-pki-complete}
$\byzantineSM$ is solvable in a fully-connected authenticated network.
\end{theorem}

\subsection{Bipartite and One-Sided Networks}\label{subsection:bipartite-pki}
The one-side and bipartite communication models offer more interesting restrictions, as described by the theorems below. 
Note that these results imply that, given PKI, $\byzantineSM$ is solvable even when one side is fully byzantine, i.e. $t_R = k$. This may seem counter-intuitive, as the honest parties' communication graph may be completely disconnected. However, we need to highlight that, if one side is completely byzantine, the $\byzantineSM$ definition allows the honest parties to simply match with nobody: any partial matching that satisfies non-competition suffices.
\begin{theorem} \label{thm:bipartite-pki-main}
$\byzantineSM$ is solvable in a bipartite authenticated network if at least one of these conditions holds: (i) $t_L, t_R < k$; (ii) $t_L < k/3$ or $t_R < k/3$.
\end{theorem}
\begin{theorem} \label{thm:one-sided-pki-main}
$\byzantineSM$ is solvable in a one-sided authenticated  network if $t_R < k$ or $t_L < k / 3$.
\end{theorem}

\paragraph{Sufficient conditions.}
In the following, we first focus on showing that the conditions described by the theorems above are sufficient.
In the unauthenticated setting, our conditions enabled us to provide the parties in the disconnected side with full communication. In turn, this allowed us to prove sufficiency by reducing the one-sided/bipartite network cases to the fully-connected network case. Assuming signatures enables us to proceed similarly: we may provide parties in side $S$ with complete communication \emph{whenever side $S'$ (which may or may not be connected) contains at least one honest party}. Roughly, the parties in $S$ will send signed messages to parties in $S'$, and the (honest) parties in $S'$ will forward these messages. Parties in $S$ will then accept messages signed correctly. We present the formal proof in Appendix \ref{appendix:bipartite-pki}.

\begin{restatable}{lemma}{PKIOneSided}\label{lemma:with-pki-one-sided}
Denote the two sides $L$ and $R$ by $S$ and $S'$.
If the parties in side $S$ are disconnected, $t_{S'} < k$, and the network is authenticated, we may assume the parties in $S$ are fully connected. 
\end{restatable}

The next corollaries follow from \cref{lemma:with-pki-one-sided}. Note that Corollary~\ref{coro:one-sided-complete} enables us to conclude that $t_R < k$ is a sufficient condition for $\byzantineSM$ in a one-sided network with PKI. Similarly, Corollary~\ref{coro:with-pki-complete} implies that $t_L, t_R < k$ is a sufficient condition for $\byzantineSM$ in a bipartite authenticated network.
\begin{corollary}\label{coro:one-sided-complete}
In a one-sided authenticated network, we may assume a fully-connected authenticated network if $t_R < k$.
\end{corollary}

\begin{corollary}\label{coro:with-pki-complete}
In a bipartite authenticated network, we may assume a fully-connected authenticated network if $t_L, t_R < k$.
\end{corollary}
Note that, for the one-sided network case we may already conclude that the conditions presented in \cref{thm:one-sided-pki-main} are sufficient. The condition $t_R < k$ follows from Corollary \ref{coro:one-sided-complete}, which enables us to reduce this case to the fully-connected authenticated setting described by Theorem~\ref{theo:with-pki-complete}. Moreover, the condition $t_L < k / 3$ is sufficient due to \cref{theorem:main:no-pki-one-sided}, which states that  $\byzantineSM$ can be solved even in an unauthenticated one-sided network when $t_L < k / 3$.

For the bipartite network case, the condition $t_L < k$ and $t_R < k$ being sufficient follows directly from Corollary \ref{coro:with-pki-complete}. Showing that having $t_L < k/3$ or $t_R < k/3$ is also sufficient introduces, however, different challenges. From this point on, we may assume without loss of generality that $t_L < k / 3$, which implies that the side $R$ may be fully byzantine. Note that this may cause the honest parties in $L$ to be completely disconnected.

While we cannot assume that the parties in $L$ are in a fully-connected network, we will be able to assume that they are \emph{in a fully-connected network with omissions}: a message may either be received within $2 \cdot \Delta$ units of time, or it is never delivered. Moreover, omissions occur \emph{only} if all parties in side $R$ are byzantine. We achieve this using the following strategy: whenever a party in $P \in L$ needs to send a message $\msg$ to a party in $P' \in R$, it sends the \emph{signed} message $\msg' := (P \rightarrow P', \timestamp, \msgId, \msg)$ to all parties in $R$, where $\msgId$ is a message identifier and $\timestamp$ is the time when $\msg'$ is sent. The (honest) parties in $R$ forward the signed message $\msg'$ to $P'$. $P'$ accepts this message only if the signature is valid and at most $2 \cdot \Delta$ time has passed since time $\tau$. We add that, as byzantine parties cannot forge signatures on the honest parties' behalf, this ensures reliable communication (up to omissions).

We may then design a protocol in the bipartite network case as follows: we (attempt) to provide the parties in $L$ with all parties' preferences lists. The parties in $L$ will run $\GaleShapley$ locally, which enables them to obtain their own matches and inform the parties in $R$ about their matches. Due to the forwarding mechanism, we may assume that the parties in $L$ are in a fully-connected network where omissions only occur if all parties in $R$ are byzantine.

To provide the honest parties in $L$ with identical views over the preferences' list, we rely on two building blocks: a synchronous $\bb$ protocol $\Pi_{\bb}$, and a synchronous \emph{Byzantine Agreement} ($\ba$) protocol $\Pi_{\ba}$. We recall the definition of $\ba$ below. 
\begin{definition}[Byzantine Agreement]\label{def:ba}
	Let $\Pi$ be a protocol where every party holds a value as input. 
    We say that $\Pi$ achieves $\ba$ if the following  hold even when up to $t$ parties are corrupted:
    \begin{itemize}[nosep]
    \item \emph{(Termination)} All honest parties output and terminate;
    \item \emph{(Validity)} If all honest parties hold the same input value $v$, they output $v$.
    \item \emph{(Agreement)} All honest parties output the same value.
    \end{itemize}
\end{definition}

The potential for omissions will require us to enhance these protocols by adding a few properties when omissions occur: termination, and \emph{weak agreement}, described below. Note that we do not require any validity condition.

\vspace{0.1cm}
\noindent \emph{(Weak agreement)}: If $P$ and $P'$ are honest and output $v \neq \bot$ and $v' \neq \bot$ respectively, $v = v'$.
\vspace{0.1cm}


The theorems below describe our building blocks. $\Pi_{\ba}$ is obtained by making adjustments to the protocol of \cite{King}, and $\Pi_{\bb}$ is a simple reduction to $\Pi_{\ba}$. For constructions, see Appendix~\ref{appendix:sync-and-omission}. 
\begin{restatable}{theorem}{BAWithOmissions}\label{thm:ba-omissions}
    Assume the $k$ parties in $L$ are in a fully-connected synchronous network with delay $\Delta$. If $t_L < k/3$, there is a $k$-party protocol $\Pi_{\ba}$ achieving $\ba$ within $\Delta_{\ba}( \Delta)$ time. Moreover, if omissions occur, $\Pi_{\ba}$ still achieves weak agreement and termination within $\Delta_{\ba}(\Delta)$ time.
\end{restatable}
\begin{restatable}{theorem}{BBWithOmissions}\label{thm:bb-omissions}
Assume the $k$ parties in $L$ are in a fully-connected synchronous network with delay $\Delta$. If $t_L < k/3$, there is a $k$-party protocol $\Pi_{\bb}$ achieving $\bb$ within $\Delta_{\bb}(\Delta)$ time. Moreover, if omissions occur, $\Pi_{\bb}$ still achieves weak agreement and termination within $\Delta_{\bb}(\Delta)$ time.
\end{restatable}

We present the code of our protocol below.

\begin{protocolbox}{$\Pi_{\byzantineSM}$}
\algoHead{Code for party $P \in R$ with input $\preferences$}
\begin{algorithmic}[1]
\State Whenever you receive a properly signed message $\msg' = (P'' \rightarrow P', \timestamp, \msgId, \msg)$ from $P'' \in L$, forward the signed message to $P' \in L$.
\State Send your preference list $\sigma$ to every party in $L$.
\State At time $\max(\Delta_{\ba}(2\Delta) + \Delta, \Delta_{\bb}(2\Delta)) + \Delta$:
\State \hspace{0.5cm} $M_p :=$ matching suggestions received from parties in $L$.
\State \hspace{0.5cm} Decide to match according to the most common suggestion \Statex \hspace{0.5cm} in $M_p$ (breaking ties arbitrarily). 
\end{algorithmic}
\algoHead{Code for party $P \in L$ with input $\preferences$}
\begin{algorithmic}[1]
\State $\msgId := 0$. Whenever you need to send a message $\msg$ to $P' \in L$, let $\timestamp :=$ the current time. Send the signed message $(P \rightarrow P', \timestamp, \msgId, \msg)$ to all parties in $R$ and increment $\msgId$.
\State In parallel:
\State \hspace{0.5cm} Send $\pi$ to all parties via $\Pi_{\bb}$. Let $\sigma_\ell$ denote the list received \Statex \hspace{0.5cm} via $\Pi_{\bb}$ from party $P_{\ell} \in L$.
\State \hspace{0.5cm} Wait $\Delta$ time to receive preference lists from parties in $R$. 
\Statex \hspace{0.5cm} Join an invocation of $\Pi_{\ba}$ for every party $P_r$ in $R$: with 
\Statex \hspace{0.5cm} input $\preferences_r$ if you have received $\preferences_r$ from $P_r$, and with a default
\Statex \hspace{0.5cm} preference list otherwise. Obtain outputs $\sigma_r$.
\State At time $\max(\Delta_{\ba}(2 \Delta) + \Delta, \Delta_{\bb}(2 \Delta))$:
\State \hspace{0.5cm} If any value in $(\sigma_v)_{v \in L \cup R}$ is $\bot$:
\State \hspace{1cm} Decide to match with nobody and terminate.
\State \hspace{0.5cm} Run $\GaleShapley$ locally with input $\left( (\sigma_l)_{l \in L}, (\sigma_r)_{r \in R} \right)$, and
\Statex \hspace{0.5cm} obtain output $M$.
\State \hspace{0.5cm} Send to each party $P_r \in R$ whom they should match to
\Statex \hspace{0.5cm} according to $M$.
\State \hspace{0.5cm} Decide who to match to according to $M$.
\end{algorithmic}
\end{protocolbox}

The next lemma states the guarantees of $\Pi_{\byzantineSM}$.
\begin{lemma}\label{lemma:protocol-omissions}
    $\Pi_{\byzantineSM}$ achieves $\byzantineSM$ in a  bipartite authenticated network if $t_L < k/3$.
\end{lemma}

We split the proof of \cref{lemma:protocol-omissions} into three lemmas. First, \cref{lemma:magic-omissions-new} describes the communication among the parties in $L$ in $\Pi_{\byzantineSM}$, allowing us to assume that the parties in $L$ are in a fully-connected network where omissions may only occur if all parties in $R$ are byzantine. Under this assumption, \cref{lemma:bdsm-protocol:omissions} shows that $\byzantineSM$ is achieved when no party in $R$ is honest, and \cref{lemma:bdsm-protocol:no-omissions} shows that $\byzantineSM$ is achieved when $R$ contains at least one honest party.
The proof of the next lemma is enclosed in Appendix  \ref{appendix:bipartite-pki}.
\begin{restatable}{lemma}{MagicOmissionsNew}\label{lemma:magic-omissions-new}
We may assume the parties in $L$ are in a fully-connected network with maximum delay $2 \cdot \Delta$  where omissions occur only if all parties in $R$ are byzantine.
\end{restatable}

\begin{lemma}\label{lemma:bdsm-protocol:omissions}
    If every party in $R$ is byzantine and $t_L < k / 3$, $\Pi_{\byzantineSM}$ achieves $\byzantineSM$.
\end{lemma}
\begin{proof}
According to \cref{lemma:magic-omissions-new}, the parties in $L$ run $\Pi_{\ba}$ and $\Pi_{\bb}$ run in a fully-connected network with omissions. As a consequence, the weak agreement and termination properties hold according to \cref{thm:ba-omissions} and \cref{thm:bb-omissions}: if the parties receive non-$\bot$ outputs, then these outputs are consistent.

Since $\Pi_{\ba}$ and $\Pi_{\bb}$ achieve termination, $\Pi_{\byzantineSM}$ achieves termination as well.


Because all parties in $R$ are byzantine, symmetry and stability are immediate: these properties concern two honest parties on opposite sides, which never happens here because one side is fully byzantine. We note that some honest parties in $L$ may have received $\bot$ and decided to match with nobody but this still results in a stable matching for this specific setting.

 As for non-competition, weak agreement guarantees that the honest parties who have obtained preference lists in each of the $\Pi_{\ba}$ and $\Pi_{\bb}$ invocations run $\GaleShapley$ with the same input. \cref{theorem:gale-shapley} ensures that these parties obtain the same matching $M$. Therefore, we conclude that $\Pi_{\byzantineSM}$ achieves $\byzantineSM$ whenever all parties in $R$ are byzantine. 


\end{proof}

\begin{lemma}\label{lemma:bdsm-protocol:no-omissions}
    If $R$ contains an honest party and $t_L < k/3$, $\Pi_\byzantineSM$ achieves $\byzantineSM$.
\end{lemma}
\begin{proof}
Parties in $L$ are in a fully-connected network with no omissions according to \cref{lemma:magic-omissions-new}, hence $\Pi_{\ba}$ and $\Pi_{\bb}$ achieve respectively $\ba$ according to \cref{thm:ba-omissions} and $\bb$ according to Theorem~\ref{thm:bb-omissions}.
Consequently, all honest parties run $\GaleShapley$ locally on the same input due to agreement and termination. Moreover, the validity properties of $\Pi_{\ba}$ and $\Pi_{\bb}$ ensure that honest parties' preference lists are received correctly and used as input in the local run of $\GaleShapley$.
We therefore obtain that the honest parties in $L$ run the same instance of $\GaleShapley$ locally, and every honest party's preference list in $\GaleShapley$ is the same as its original input. Therefore, the stable matching $M$ computed by $\GaleShapley$ also satisfies our $\byzantineSM$ definition.

The last step is to prove that every honest party decides according to $M$. This is immediate for honest parties in $L$. As for honest parties in $R$, they decide according to the most common option sent by parties in $L$. Since $k - t_L > t_L$ of the parties in $L$ are honest, each party in $R$ receives its match in $M$ as the majority option, and decides on this match.
Therefore, all honest parties decide according to $M$. Consequently $\Pi_{\byzantineSM}$ achieves $\byzantineSM$ whenever $R$ contains at least one honest party.
\end{proof}


\paragraph{Necessary Conditions.} We still need to show that the conditions presented in \cref{thm:bipartite-pki-main} and \cref{thm:one-sided-pki-main} are necessary.
We write our proof for one-sided communication, and the bipartite network case will be a corollary.
\begin{lemma}\label{lemma:with-pki-one-sided-impossible}
If $t_R = k$ and $t_L \geq k/3$, then achieving $\byzantineSM$ is a one-sided network is impossible.
\end{lemma}
\begin{proof}
We assume by contradiction that there is a protocol achieving $\byzantineSM$ in this setting. Using Corollary \ref{coro:to-simplified} and \cref{lemma:reduce-number}, this means that there exists a protocol $\Pi$ which solves $\simplifiedSM$ on $n := 6$ nodes with $t_R = 3$ and $t_L = 1$. We denote the six parties by  $L = \{a,b,c\}$ and $R = \{u,v,w\}$, and assume that $b$ and all parties in $R$ are byzantine. In the following, we fix an input configuration, and we describe an adversarial strategy that breaks the guarantees of $\Pi$ in this setting: honest parties $a$ and $c$ will match with the same byzantine party $v$ in $R$, hence breaking non-competition.

\begin{figure}[h]
\centering
\includegraphics[scale=0.55]{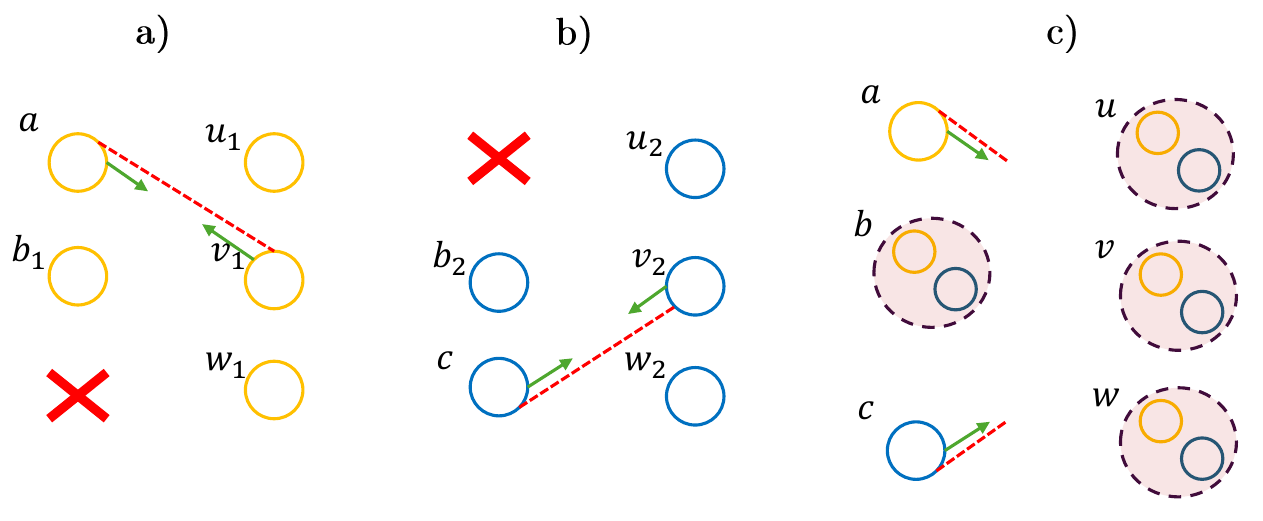}
\caption{a) From the point of view of $a$, all parties are honest except $c$ (which crashed), simplified stability guarantees that $a$ matches with $v$. b) From the point of view of $c$, all parties are honest except $a$ (which crashed), simplified stability guarantees that $c$ matches with $v$. c) What is actually happening is that all byzantine parties are simulating two versions of themselves except $a$ and $c$, but both $a$ and $c$ are honest and try to match $v$, which is not allowed by the non-competition property.}

\end{figure}

To do so, we define $a$ and $c$'s favorite as $v$. Moreover, each byzantine party will internally simulate two instances of themselves running protocol $\Pi$. For example, $v$ will internally simulate two instances of itself $v_1$ and $v_2$ such that $v_1$'s favorite is $a$ and $v_2$'s favorite is $b$. Each of the remaining byzantine parties $x \in \{ b, u, w \}$ simulates two instances of itself $x_1$ and $x_2$ with any arbitrary inputs.


Each communication edge has at least one of its two endpoints being a byzantine party in our setting: since the parties in $L$ are only connected to parties in $R$ in a one-sided network, byzantine parties have full control over the communication network. They may therefore divide the communication network in two groups: $\{a,b_1, u_1, v_1, w_1\}$ and $\{b_2, c, u_2, v_2, w_2\}$.
Messages only get sent and received within a group, meaning that if party $a$ running $\Pi$ wants to send a message to $w$, $w_1$ will receive it. Any message sent from $u_1, v_1$ and $w_1$ to $c$ is never received. 

We consider the output of $a$ and $c$ after running protocol $\Pi$ in this setting. We have $t_L = 1$ and $t_R = 3$, which satisfies $\Pi$'s requirements. Consequently, termination holds: $a$ and $c$ must decide to either match with some party or no one.

We then consider a new scenario for party $a$: parties $a, b, u, v, w$ are all honest with the same favorites as the previous scenario's first group. Party $c$ is byzantine and crashes at the beginning, i.e., it does not send any message. In this scenario, $t_L = 1$ and $t_R = 0$, which satisfies the requirements of $\Pi$. Therefore, termination holds and $a$ obtains an output. Since both $a$ and $v$ are honest and each other's favorite, they must match with each other according to simplified stability. However, we remark that $a$ cannot distinguish between this scenario and the previous one: it receives the exact same messages in both cases. Therefore, in the first scenario, $a$ also decides to match with $v$.

We may construct a symmetric scenario for party $c$: this time, parties $b, c, u, v, w$ are honest with the same inputs as in the first scenario's second group. Party $a$ is byzantine and crashes at the beginning of the protocol's execution. As $t_L = 1$ and $t_R = 0$, the requirements of $\Pi$ are satisfied: termination and simplified stability hold. Therefore, $c$ outputs $v$. Moreover, this scenario is indistinguishable to $c$ from the first scenario, hence $c$ matches with $v$ in the first scenario as well.


We consequently obtain a contradiction: in the first scenario, both $a$ and $c$ are honest and match with the same party, which breaks non-competition.
\end{proof}

As the bipartite communication model is weaker than the one-sided model, \cref{lemma:with-pki-one-sided-impossible} provides the following corollary.
\begin{corollary}\label{corollary:with-pki-bipartite-impossible}
If $t_R = k$ (resp. $t_L = k$) and $t_L \geq k/3$ (resp. $t_R \geq k/3$), then achieving $\byzantineSM$ in a bipartite network is impossible.
\end{corollary}


\paragraph{Putting it all together.} We conclude the section by providing the formal proofs of \cref{thm:bipartite-pki-main} and \cref{thm:one-sided-pki-main}. 
We first present the proof of \cref{thm:bipartite-pki-main}, focusing on a bipartite network.


\begin{proof}[Proof of \cref{thm:bipartite-pki-main}]
We first discuss sufficiency.
If $t_L < k$ and $t_R < k$, Corollary \ref{coro:with-pki-complete} enables us to assume a fully-connected network. Therefore, using \cref{theo:with-pki-complete}, we obtain that $\byzantineSM$ is solvable. If $t_L < k / 3$ and $t_R \leq k$, \cref{lemma:protocol-omissions} describes a protocol achieving $\byzantineSM$. The case $t_R < k / 3$ and $t_L \leq k$ is symmetrical.


Otherwise, if $t_L \geq k/3$ and $t_R = k$ or the opposite, we may apply Corollary \ref{corollary:with-pki-bipartite-impossible} and conclude that $\byzantineSM$ is impossible. 
\end{proof}

We now prove \cref{thm:one-sided-pki-main}, discussing the one-sided network case.
\begin{proof}[Proof of \cref{thm:one-sided-pki-main}]
For sufficiency, when $t_R < k$, we may apply \cref{lemma:with-pki-one-sided} and hence assume a fully-connected network. Then, \cref{theo:with-pki-complete} enables us to conclude that $\byzantineSM$ is solvable. If $t_R = k$ and $t_L < k/3$, \cref{theorem:main:no-pki-one-sided} guarantees that $\byzantineSM$ is solvable.
    
When none of these conditions holds, i.e., if $t_R = k$ and if $t_L \geq k/3$, \cref{lemma:with-pki-one-sided-impossible} enables us to conclude that $\byzantineSM$ is impossible.
\end{proof}

\section{Conclusion}

We investigated whether stable matching can be achieved in a synchronous network where some of the parties involved may be byzantine. We analyzed this problem under various network topologies, both with and without cryptographic assumptions. For each setting, we gave necessary and sufficient conditions,  assuming that each party holds as input a complete ranking of the parties on the other side.


Our work highlights multiple promising directions for further research. A first direction could be generalizing our results to the \emph{stable roommate} problem. Instead of assuming that the parties to be matched are in two disjoint sets, the stable roommate problem seeks a stable matching within the same set. Note that our necessary conditions also apply to a byzantine variant of the stable roommate problem, even though there is no longer a distinction between byzantine parties on the two sides. However, the stable matching problem comes with the guarantee that a stable matching always exists, while the stable roommate problem does not. Hence, definitions and properties need to be refined to account for this.

Another interesting direction would be to extend our question to the asynchronous model. Using our current definitions, one can prove that even if only one party known in advance can be byzantine, stable matching is not solvable. Therefore, the properties required for the stable matching would have to be relaxed for this problem to be of interest.

Finally, while our work has provided a complete characterization in terms of solvability, there are multiple aspects in which our feasibility results could be improved. This includes improvements in terms of efficiency (i.e., communication complexity), but also improvements in terms of guarantees, such as providing some degree of privacy.



\newpage
\appendix
\section{Appendix}
\subsection{Preliminaries: Missing Proofs}\label{appendix:preliminaries}
We present the formal proof of Lemma \ref{lemma:broadcast-easy}, establishing that whenever $\bb$ can be achieved, the $\byzantineSM$ problem is solvable.
\BroadcastEasy*

\begin{proof}
The parties distribute their input preference lists via $\bb$. This provides the parties with an identical view over the parties' preferences lists. If a party $P$ has not sent a valid preference list, then $P$ is byzantine, and the honest parties may simply assign a pre-defined default preference list to it. 
Afterwards, each party runs $\GaleShapley$ offline with the preference lists obtained and obtains a matching $M$. Each party then outputs its match in $M$.

Termination comes from $\bb$'s termination property. $\bb$'s validity condition ensures that, if the input of an honest party $P$ in our $\byzantineSM$ instance is the preference list $\pi_P$, then party $P$ has preference list $\pi_P$ in each of the honest parties' offline executions of $\GaleShapley$.
$\bb$'s Agreement properties ensures that all honest parties run $\GaleShapley$ with the same input. Since $\GaleShapley$ is deterministic, all honest parties obtain the same output $M$. 

According to Theorem \ref{theorem:gale-shapley}, $M$ is a proper matching (if $u$ is matched with $v$, then $v$ is matched with $u$) satisfying stability (no blocking pair). Therefore, as each honest party outputs its match in $M$, symmetry, non-competition and stability hold. Hence, $\byzantineSM$ is achieved.
\end{proof}

\subsection{Simplified Stable Matching: Missing Proofs} \label{appendix:simplified-stable-matching}

We first include the proof of Lemma \ref{coro:to-simplified}, establishing that $\simplifiedSM$ reduces to $\byzantineSM$.

\SimplifiedReduction*

\begin{proof}
It suffices to show that any protocol solving $\byzantineSM$ also solves $\simplifiedSM$. 
Given a protocol $\Pi$ solving $\byzantineSM$, we construct a protocol $\Pi'$ that solves $\simplifiedSM$, as follows:

Given its favorite as input, each party constructs an arbitrary preference list with the favorite ranked first. Afterward, parties join an invocation of $\Pi$ with the constructed lists as inputs. The output obtained in $\Pi$ for $\byzantineSM$ is used in $\Pi'$ as the output for $\simplifiedSM$.


First, $\Pi'$ maintains the resilience thresholds of $\Pi$. Moreover, the termination, symmetry, and non-competition guarantees of $\Pi'$ follow directly from $\Pi$ achieving termination, symmetry, and non-competition, respectively. Finally, if two honest parties are each other's favorites, they rank each other first in the constructed lists and, consequently, always form a blocking pair if they are not matched. Therefore, the simplified stability property of $\Pi'$ is guaranteed by the stability property of $\Pi$.
\end{proof}

We now present the proof of Lemma \ref{lemma:reduce-number}, allowing us to extend impossibility results from small settings to larger settings.
\ReduceNumberLemma*
\begin{proof}
We partition $L$ into $d$ disjoint sets $L_1$, \dots, $L_d$ such that $1 \leq |L_1|,\dots,|L_d| \leq \lceil |L| / d \rceil = \lceil k/d \rceil$. Similarly, we partition $R$ into $d$ disjoint sets $R_1$, \dots, $R_d$ such that $1 \leq |R_1|,\dots,|R_d| \leq \lceil |R| / d \rceil = \lceil k/d \rceil$. From each of these sets, we pick one representative: $l_1, \dots, l_d, r_1, \dots, r_d$. We build $\Pi'$ solving $\simplifiedSM$ for $2d$ parties: $l_1', \dots, l_d'$ on the left side and $r_1', \dots, r_d'$ on the right side, as follows:
\begin{itemize}[nosep]
\item Each party $l_i'$ in $\Pi'$ simulates all the parties in $L_i$ running $\Pi$. Similarly, each party $r_j'$ in $\Pi'$ simulates all the parties in $R_j$ running $\Pi$.
\item Input: If the input (favorite) of $l_i'$ is $r_j'$, then we assign $r_j$ as the favorite of $l_i$. Similarly, if the input of $r_j'$ is $l_i'$, then we assign $l_i$ as the favorite of $r_j$. For parties that are not representatives of their group, we assign arbitrary favorites.
\item Output: For a given $l_i$, if there is a $r_j$ such that $l_i$ matches $r_j$, then $l_i'$ declares that it matches $r_j'$. Otherwise $l_i'$ declares that it matches nobody.
\end{itemize}

We are essentially running the $\simplifiedSM$ algorithm on the whole graph, but only looking at the representative of each set and discarding anything unrelated to them. As a consequence, $\Pi'$ achieves termination, symmetry, simplified stability, and non-competition since $\Pi$ achieves termination, symmetry, simplified stability, and non-competition.

As each party in $\Pi'$ simulates up to $\lceil k / d \rceil$ parties 
from $\Pi$ and $\Pi$ supports up to $t_L$ byzantine parties in $L$ and $t_R$ byzantine parties in $R$, the bound on the number of byzantine parties supported by $\Pi'$ follows immediately.
\end{proof}

\subsection{Byzantine Broadcast with General Adversaries}\label{appendix:general-adversaries}
To achieve the feasibility part of Theorem \ref{theo:pki-complete}, 
we got help form the result below.

\GeneralAdversaries*

As mentioned before, this is a corollary of \cite[{Theorem 2}]{DISC:FitMau98}. We again highlight that \cite{DISC:FitMau98} assumes a \emph{general adversary}, which we briefly introduce next. In this adversarial model, the corruption power of the adversary is specified by a (subset-closed) 
\emph{adversarial structure} $\mathcal{Z} \subseteq 2^\mathcal{P}$, where $\mathcal{P}$ denotes the set of parties. In particular, the adversary may choose to corrupt any set of parties in $\mathcal{Z}$.
For instance, if $\mathcal{P} := \{P_1, P_2, \ldots, P_5\}$, a potential adversarial structure $\mathcal{Z}$ is $\{\varnothing, \{P_1\}, \{P_2\}, \{P_1, P_2\}, \{P_4\}\}$, which means that the adversary may choose between corrupting no parties, corrupting parties $P_1, P_2$ (or only one of the two), or corrupting only party $P_4$. In contrast, one often considers a \emph{threshold adversary}, which may corrupt up to $t$ of the $n$ parties (as is the case in most literature): this is a particular case of the general adversary model where $\mathcal{Z}$ is the set of all subsets of at most $t$ parties. The adversary assumed in our work sits in-between these two: we assume that the adversary may corrupt up to $t_L$ parties in $L$ and up to $t_R$ parties in $R$, hence our adversary structure is $\mathcal{Z^\star} := \{S_L \cup S_R \mid S_L \subseteq L, S_R \subseteq R, \abs{S_L} \leq t_L, \abs{S_R} \leq t_R\}$. This can be thought of as the product of two threshold adversary structures.

In the general adversaries model, {\cite[{Theorem 2}]{DISC:FitMau98}} states the following:
\begin{theorem}[\hspace{-1pt}{\cite[{Theorem 2}]{DISC:FitMau98}}] \label{thm:general-adversaries-explicit}
    Assume a fully-connected unauthenticated network, and an adversary structure $\mathcal{Z}$ such that for any three sets $Z_1, Z_2, Z_3 \in \mathcal{Z}$ it holds that $Z_1 \cup Z_2 \cup Z_3 \neq \mathcal{P}$. Then, there is a protocol achieving $\bb$ in this setting.
\end{theorem}

Then, to prove Lemma \ref{lemma:general-adversaries}, we only need to show no three sets in our 
adversary structure $\mathcal{Z}^\star$ cover the set of $n$ parties:
\begin{proof}[Proof of Lemma \ref{lemma:general-adversaries}]
Consider our adversarial structure $\mathcal{Z^\star} := \{S_L \cup S_R \mid S_L \subseteq L, S_R \subseteq R, \abs{S_L} \leq t_L, \abs{S_R} \leq t_R\}$, and let $Z_1, Z_2, Z_3 \in \mathcal{Z^\star}$ be arbitrary. We show that $Z_1 \cup Z_2 \cup Z_3 \neq L \cup R$.

Without loss of generality, we may assume that the condition $t_L < k / 3$ holds (the case where $t_R < k / 3$ and $t_L \geq k / 3$ is analogous). As every $Z \in Z^\star$ contains at most $t_L$ parties in $L$, it follows that $Z_1 \cup Z_2 \cup Z_3$ contain at most $3 \cdot t_L < 3 \cdot k / 3 = k$ parties in $L$. Hence, at least one element of $L$ is uncovered by $Z_1 \cup Z_2 \cup Z_3$, from which $Z_1 \cup Z_2 \cup Z_3 \neq L \cup R$.

Then, we may apply Theorem \ref{thm:general-adversaries-explicit} and conclude that there is a protocol achieving $\bb$ in our setting.
\end{proof}

\subsection{Unauthenticated Setting: Missing Proofs}\label{appendix:complete-bipartite-graph-without-pki}
We present the proof of Lemma \ref{lemma:pki-bipartite}. This has provided reductions between the communication models when analyzing $\byzantineSM$ in unauthenticated settings.
\BipartiteCommunicationNoPKI*

\begin{proof}
Let $u,v$ be parties in $S$. We want to simulate an authenticated channel between $u$ and $v$, i.e the receiver knows who the sender is. 

If $u$ wants to send a message $M$ to $v$, it sends the message ($u \rightarrow v$, $M$) to every party in $S'$. Then if a party in $S'$ receives a message ($u \rightarrow v$, $M$) from $u$, it forwards it to $v$. Finally, if $v$ receives the same message ($u \rightarrow v$, $M$) from a majority (i.e strictly more than $k/2$) of $S'$, it considers it received message $M$ from $u$.

Using this strategy, we can see that sending a message takes a bounded amount of time (at most $2 \Delta$). Moreover, if $u$ is honest and sends a message $M$, at least $k - t_{S'} > k/2$ parties from $S'$ will forward it to $v$ which will therefore accept it.
If $v$ accepts a message $M$ from $u$, this means strictly more than $k/2$ parties from $S'$ forwarded it. Because $t_{S'} < k/2$, at least one honest party forwarded it, meaning $u$ intended to send this message (being honest or not).
\end{proof}

\subsection{Authenticated Setting: Missing Proofs} \label{appendix:bipartite-pki}
We present the proof of Lemma \ref{lemma:with-pki-one-sided}, which has provided us with reductions between the communicated models when analyzing $\byzantineSM$ in authenticated settings.

\PKIOneSided*

\begin{proof}
Let $u,v$ be parties in $S$. We want to simulate an authenticated channel between $u$ and $v$.

If $u$ wants to send a message $M$ to $v$, it sends the the signed message ($u \rightarrow v$, $M$) to every party in $S'$. Then if a party in $S'$ receives a message ($u \rightarrow v$, $M$) with a valid signature from $u$, it forwards it to $v$. Finally, if $v$ receives a message ($u \rightarrow v$, $M$) from a party in $S'$ with a valid signature, it considers it receives message $M$ from $u$.

Using this strategy, we can see that sending a message takes a bounded amount of time (at most $2 \Delta$). Moreover, if $u$ is honest and sends a message $M$, at least $k - t_{S'} > 0$ parties from $S'$ will forward it to $v$, which will therefore accept it.
If $v$ accepts a message $M$ from $u$, this message is signed by $u$, meaning $u$ intended to send this message (being honest or not).
\end{proof}

\MagicOmissionsNew*
\begin{proof}
Let $(u,v)$ be two parties in $L$ and assume that $u$ wants to send a message $\msg$ to $v$. $u$ sends a signed message $(u \rightarrow v, \timestamp, \msgId, \msg)$ to all parties in $R$, $\tau$ being the current timestamp and $\msgId$ being a message identifier. Parties in $R$ then forward this signed message to $v$. If $v$ receives a message $(u, \rightarrow, v, \timestamp, \msgId, \msg)$ properly signed by $u$ such that $\timestamp$ is at most $2\Delta$ units of time in the past and $u$ has not seen $\msgId$ has not been seen before, it accepts message $\msgId$ from $u$.

With this approach, if at least one party in $R$ is honest, messages always get forwarded and received within $2\Delta$ units of time. Otherwise, note that byzantine parties cannot forge signatures on the honest parties' behalf: the byzantine parties may choose whether to forward the message or not, then causing an omission.
\end{proof}

\subsection{Protocols in Settings with or without Omissions} \label{appendix:sync-and-omission}
In order to prove sufficiency when one side may be completely byzantine in Section \ref{subsection:bipartite-pki}, we have considered a setting consisting of a fully-connected synchronous network where \emph{omissions} may occur: if a message is delivered, it is delivered within $\Delta$ time. We have utilized the building blocks described by the theorems below:
\BAWithOmissions*
\BBWithOmissions*

In the following, we describe the constructions behind these theorems.

\paragraph{Byzantine Agreement.} We start by presenting protocol $\Pi_{\ba}$.
We need a \emph{synchronous} $k$-party $\ba$ protocol resilient against $t_L < k / 3$ corruptions. We may use, for instance, the protocol of \cite{King}, presented below.
\begin{protocolbox}{$\Pi_{\king}$}
    \algoHead{Code for party $P \in L$ with input $v_{\inputt}$}
    \begin{algorithmic}[1]
    \State If you have not obtained any output by time $3(t_L + 1) \cdot \Delta$, output $\bot$.
    \State $v := v_{\inputt}$
    \For {$i = 1 \ldots t_L + 1$}
    \State \textbf{(Round 1)}
    \State Send $(\val, v)$ to all parties. 
    \State Wait $\Delta$ time.
    \State \textbf{(Round 2)}
    \State If you have received $(\val, v')$ for the same $v'$ from $k - t_L$ parties in $L$:
    \State \hspace{0.5cm} Send $(\propose, v')$ to all parties
    \State Wait $\Delta$ time.
    \State \textbf{(Round 3})
    \State If you have received some $(\propose, v')$ from more than $t_L$ parties, set $v' = v$.
    \State \textbf{King $P_i$ only}: Send $v_K := v$ to all parties.
    \State Wait $\Delta$ time.
    \State If you have received strictly less than $k - t_L$ messages $(\propose, v')$ for any $v'$:
    \State \hspace{0.5cm} If you have received $v_K$ from king $P_i$: Set $v = v_K$
    \EndFor
    \State Output $v$
\end{algorithmic}
\end{protocolbox}

The next theorem comes directly from \cite{King}.
\begin{theorem}[Theorem 3.1 of \cite{King}]
    Whenever $\Pi_{\king}$ runs in a synchronous network with maximum delay $\Delta$ and at most $t_L < k / 3$ byzantine corruptions, $\Pi_{\king}$ achieves $\ba$ within $\Delta_{\king} := 3(t_L + 1) \cdot \Delta$ time.
\end{theorem}

We note that, due to line 1, termination is guaranteed even when omissions occur.
\begin{remark}
    Whenever $\Pi_{\king}$ runs in a synchronous network with omissions with maximum delay $\Delta$, it achieves termination within $\Delta_{\king}(\Delta) := 3(t_L + 1) \cdot \Delta$ time.
\end{remark}

However, $\Pi_{\king}$ does not achieve weak agreement when omissions occur. To achieve this property, we need one more round of communication, as presented below.

\begin{protocolbox}{$\Pi_{\ba}$}
    \algoHead{Code for party $P \in L$ with input $v_{\inputt}$}
    \begin{algorithmic}[1]
    \State Join $\king$ with input $v_{\inputt}$ and obtain output $y$.
    \State At time $\Delta_{\king}(\Delta)$, send $y$ to every party.
    \State If the same value $z$ is received from $k - t_L$ parties in $L$ by time $\Delta_{\king}(\Delta) + \Delta$, output $z$. Otherwise, output $\bot$.
\end{algorithmic}
\end{protocolbox}


\begin{proof}[Proof of Theorem \ref{thm:ba-omissions}]
Termination  is achieved within $\Delta_\ba(\Delta) = 3(t_L + 1) \cdot \Delta + \Delta$ time even if omissions occur: this follows from $\Pi_{\king}$'s termination guarantees. 

We first show that $\ba$ is achieved when no omissions occur. In this case $\Pi_{\king}$ achieves $\ba$. Hence, all honest parties obtain the same value $y$. Then, at least the $k - t_L$ honest parties send $y$ to all parties. Since no omissions occur, these messages are received within $\Delta$ time, and all parties output $y$. Moreover, due to $\Pi_{\king}$'s validity, if all honest parties had the same input $v$, the honest parties have obtained $y = v$, and therefore all honest parties output $y$. Consequently, $\ba$ is achieved.


We may now discuss weak agreement if omissions occur. In this case $\Pi_{\king}$ only achieves termination (and it is possible that the honest parties obtain $y = \bot$). Assume that an honest party $p$ outputs $z \neq \bot$ in $\Pi_{\ba}$. Then, $p$ has received $z$ from $k - t_L$ parties, hence from at least $n - 2t_L > t_L$ honest parties. As such, every party can receive strictly less than $n - t_L$ values $z' \neq z$, and therefore no honest party outputs $z' \neq z$.  Therefore, weak agreement holds.

\end{proof}

\paragraph{Byzantine Broadcast.}
Protocol $\Pi_{\bb}$ is a simple reduction to $\Pi_{\ba}$. The sender $S$ sends its value to all parties, and afterwards the parties run $\Pi_{\ba}$ to agree on the value received.
\begin{protocolbox}{$\Pi_{\bb}$}
    \algoHead{Code for sender $S \in L$ with input $v_{S}$}
    \begin{algorithmic}[1]
    \State Send $v_{S}$ to all parties.
    \end{algorithmic}

    \algoHead{Code for party $P \in L$}
    \begin{algorithmic}[1]
    \State Receive value $v$ from the sender. If you did not receive a value within $\Delta$ time, set $v :=$ default value (default preference list).
    \State At time $\Delta$, join $\Pi_{\ba}$ with input $v$. Return the output obtained.
\end{algorithmic}
\end{protocolbox}


\begin{proof}[Proof of Theorem \ref{thm:bb-omissions}]
$\Pi_{\ba}$ achieves termination within $\Delta_{\ba}(\Delta)$ time even when omissions occur, hence $\Pi_{\bb}$ achieves termination within $\Delta_{\bb}(\Delta) = \Delta + \Delta_{\ba}(\Delta)$ time even when omissions occur.

If no omissions occur, $\Pi_{\ba}$ achieves $\ba$. Therefore, $\Pi_{\bb}$ achieves agreement. If the sender is honest with input $v_S$, all honest parties join $\Pi_{\ba}$ with input $v_S$, and the validity guarantee of $\Pi_{\ba}$ ensures that the parties output $v_S$. Then, $\Pi_{\bb}$ achieves validity, and consequently $\bb$.

Lastly, if omissions occur, as $\Pi_{\ba}$ achieves weak agreement, $\Pi_{\bb}$ also achieves weak agreement.
\end{proof}


\newpage
\begin{thebibliography}{10}

\bibitem{ShapleyRoth}
Stable allocations and the practice of market design: The royal swedish academy of sciences.
\newblock {\em The Indian Economic Journal}, 60(4):3--34, 2013.
\newblock \href {https://doi.org/10.1177/0019466220130402} {\path{doi:10.1177/0019466220130402}}.

\bibitem{BaLoHa12}
Siavash Bayat, Raymond H.~Y. Louie, Zhu Han, Yonghui Li, and Branka Vucetic.
\newblock Distributed stable matching algorithm for physical layer security with multiple source-destination pairs and jammer nodes.
\newblock pages 2688--2693, 2012.
\newblock \href {https://doi.org/10.1109/WCNC.2012.6214256} {\path{doi:10.1109/WCNC.2012.6214256}}.

\bibitem{BaLoLi11}
Siavash Bayat, Raymond H.~Y. Louie, Yonghui Li, and Branka Vucetic.
\newblock Cognitive radio relay networks with multiple primary and secondary users: Distributed stable matching algorithms for spectrum access.
\newblock In {\em 2011 IEEE International Conference on Communications (ICC)}, pages 1--6, 2011.
\newblock \href {https://doi.org/10.1109/icc.2011.5962935} {\path{doi:10.1109/icc.2011.5962935}}.

\bibitem{King}
Piotr Berman, Juan~A. Garay, and Kenneth~J. Perry.
\newblock Towards optimal distributed consensus.
\newblock In {\em Proceedings of the 30th Annual Symposium on Foundations of Computer Science (FOCS)}, pages 410--415. IEEE Computer Society, 1989.
\newblock \href {https://doi.org/10.1109/SFCS.1989.63511} {\path{doi:10.1109/SFCS.1989.63511}}.

\bibitem{ChHiSe02}
Subhendu Chattopadhyay, Lisa Higham, and Karen Seyffarth.
\newblock Dynamic and self-stabilizing distributed matching.
\newblock In {\em Proceedings of the Twenty-First Annual Symposium on Principles of Distributed Computing}, PODC '02, page 290–297, New York, NY, USA, 2002. Association for Computing Machinery.
\newblock \href {https://doi.org/10.1145/571825.571877} {\path{doi:10.1145/571825.571877}}.

\bibitem{DolStr83}
Danny Dolev and H.~Raymond Strong.
\newblock Authenticated algorithms for byzantine agreement.
\newblock {\em SIAM Journal on Computing}, 12(4):656--666, 1983.
\newblock \href {https://doi.org/10.1137/0212045} {\path{doi:10.1137/0212045}}.

\bibitem{ElAhDa12}
Ahmad~M. El-Hajj, Zaher Dawy, and Walid Saad.
\newblock A stable matching game for joint uplink/downlink resource allocation in ofdma wireless networks.
\newblock In {\em 2012 IEEE International Conference on Communications (ICC)}, pages 5354--5359, 2012.
\newblock \href {https://doi.org/10.1109/ICC.2012.6364329} {\path{doi:10.1109/ICC.2012.6364329}}.

\bibitem{PODC:FisLynMer85}
Michael~J. Fischer, Nancy~A. Lynch, and Michael Merritt.
\newblock Easy impossibility proofs for distributed consensus problems.
\newblock In Michael~A. Malcolm and H.~Raymond Strong, editors, {\em 4th ACM PODC}, pages 59--70. {ACM}, August 1985.
\newblock \href {https://doi.org/10.1145/323596.323602} {\path{doi:10.1145/323596.323602}}.

\bibitem{DISC:FitMau98}
Matthias Fitzi and Ueli~M. Maurer.
\newblock Efficient byzantine agreement secure against general adversaries.
\newblock In {\em Proceedings of the 12th International Symposium on Distributed Computing}, DISC '98, page 134–148, Berlin, Heidelberg, 1998. Springer-Verlag.

\bibitem{GayleShapley}
D.~Gale and L.~S. Shapley.
\newblock College admissions and the stability of marriage.
\newblock {\em The American Mathematical Monthly}, 69(1):9--15, 1962.
\newblock URL: \url{http://www.jstor.org/stable/2312726}.

\bibitem{GONCZAROWSKI2019626}
Yannai~A. Gonczarowski, Noam Nisan, Rafail Ostrovsky, and Will Rosenbaum.
\newblock A stable marriage requires communication.
\newblock {\em Games and Economic Behavior}, 118:626--647, 2019.
\newblock \href {https://doi.org/10.1016/j.geb.2018.10.013} {\path{doi:10.1016/j.geb.2018.10.013}}.

\bibitem{GuZhPa15}
Yunan Gu, Yanru Zhang, Miao Pan, and Zhu Han.
\newblock Matching and cheating in device to device communications underlying cellular networks.
\newblock {\em IEEE Journal on Selected Areas in Communications}, 33(10):2156--2166, 2015.
\newblock \href {https://doi.org/10.1109/JSAC.2015.2435361} {\path{doi:10.1109/JSAC.2015.2435361}}.

\bibitem{GusfieldIrving}
Dan Gusfield and Robert~W. Irving.
\newblock {\em The stable marriage problem: structure and algorithms}.
\newblock MIT Press, Cambridge, MA, USA, 1989.

\bibitem{HaKa09}
Rachid Hadid and Mehmet Karaata.
\newblock Stabilizing maximum matching in bipartite networks.
\newblock {\em Computing}, 84:121--138, 04 2009.
\newblock \href {https://doi.org/10.1007/s00607-009-0025-z} {\path{doi:10.1007/s00607-009-0025-z}}.

\bibitem{HSU199277}
Su-Chu Hsu and Shing-Tsaan Huang.
\newblock A self-stabilizing algorithm for maximal matching.
\newblock {\em Information Processing Letters}, 43(2):77--81, 1992.
\newblock \href {https://doi.org/10.1016/0020-0190(92)90015-N} {\path{doi:10.1016/0020-0190(92)90015-N}}.

\bibitem{HuChi06}
Chien-Chung Huang.
\newblock Cheating by men in the gale-shapley stable matching algorithm.
\newblock In {\em Algorithms -- ESA 2006}, pages 418--431. Springer Berlin Heidelberg, 2006.
\newblock \href {https://doi.org/10.1007/11841036_39} {\path{doi:10.1007/11841036_39}}.

\bibitem{Khanchandani2016DistributedSM}
Pankaj Khanchandani and Roger Wattenhofer.
\newblock {Distributed Stable Matching with Similar Preference Lists}.
\newblock In {\em 20th International Conference on Principles of Distributed Systems (OPODIS 2016)}, pages 12:1--12:16, 2017.
\newblock \href {https://doi.org/10.4230/LIPIcs.OPODIS.2016.12} {\path{doi:10.4230/LIPIcs.OPODIS.2016.12}}.

\bibitem{Kipnis}
Alex Kipnis and Boaz Patt-Shamir.
\newblock Brief announcement: a note on distributed stable matching.
\newblock In {\em Proceedings of the 28th ACM Symposium on Principles of Distributed Computing}, PODC '09, page 282–283, New York, NY, USA, 2009. Association for Computing Machinery.
\newblock \href {https://doi.org/10.1145/1582716.1582766} {\path{doi:10.1145/1582716.1582766}}.

\bibitem{LSP82}
Leslie Lamport, Robert Shostak, and Marshall Pease.
\newblock The byzantine generals problem.
\newblock {\em ACM Transactions on Programming Languages and Systems}, 4(3):382--401, 1982.
\newblock \href {https://doi.org/10.1145/357172.357176} {\path{doi:10.1145/357172.357176}}.

\bibitem{LaMa17}
Marie Laveau, George Manoussakis, Joffroy Beauquier, Thibault Bernard, Janna Burman, Johanne Cohen, and Laurence Pilard.
\newblock Self-stabilizing distributed stable marriage.
\newblock In Paul Spirakis and Philippas Tsigas, editors, {\em Stabilization, Safety, and Security of Distributed Systems}, pages 46--61, Cham, 2017. Springer International Publishing.

\bibitem{MaSi15}
Bruce~M. Maggs and Ramesh~K. Sitaraman.
\newblock Algorithmic nuggets in content delivery.
\newblock {\em SIGCOMM Comput. Commun. Rev.}, 45(3):52–66, July 2015.
\newblock \href {https://doi.org/10.1145/2805789.2805800} {\path{doi:10.1145/2805789.2805800}}.

\bibitem{MaTu18}
Tung Mai and Vijay~V. Vazirani.
\newblock {Finding Stable Matchings That Are Robust to Errors in the Input}.
\newblock In {\em 26th Annual European Symposium on Algorithms (ESA 2018)}, volume 112 of {\em Leibniz International Proceedings in Informatics (LIPIcs)}, pages 60:1--60:11, Dagstuhl, Germany, 2018. Schloss Dagstuhl -- Leibniz-Zentrum f{\"u}r Informatik.
\newblock \href {https://doi.org/10.4230/LIPIcs.ESA.2018.60} {\path{doi:10.4230/LIPIcs.ESA.2018.60}}.

\bibitem{MaMjPi07}
Fredrik Manne, Morten Mjelde, Laurence Pilard, and S{\'e}bastien Tixeuil.
\newblock A new self-stabilizing maximal matching algorithm.
\newblock In {\em Structural Information and Communication Complexity}, pages 96--108, Berlin, Heidelberg, 2007. Springer Berlin Heidelberg.
\newblock \href {https://doi.org/10.1016/j.tcs.2008.12.022} {\path{doi:10.1016/j.tcs.2008.12.022}}.

\bibitem{Ostrovsky}
Rafail Ostrovsky and Will Rosenbaum.
\newblock Fast distributed almost stable matchings.
\newblock In {\em Proceedings of the 2015 ACM Symposium on Principles of Distributed Computing}, PODC '15, page 101–108, New York, NY, USA, 2015. Association for Computing Machinery.
\newblock \href {https://doi.org/10.1145/2767386.2767424} {\path{doi:10.1145/2767386.2767424}}.

\bibitem{PaBeSa13}
Francesco Pantisano, Mehdi Bennis, Walid Saad, Stefan Valentin, and Mérouane Debbah.
\newblock Matching with externalities for context-aware user-cell association in small cell networks.
\newblock In {\em 2013 IEEE Global Communications Conference (GLOBECOM)}, pages 4483--4488, 2013.
\newblock \href {https://doi.org/10.1109/GLOCOMW.2013.6855657} {\path{doi:10.1109/GLOCOMW.2013.6855657}}.

\bibitem{Roth1982TheEO}
Alvin~E. Roth.
\newblock The economics of matching: Stability and incentives.
\newblock {\em Math. Oper. Res.}, 7:617--628, 1982.
\newblock \href {https://doi.org/10.1287/moor.7.4.617} {\path{doi:10.1287/moor.7.4.617}}.

\end{thebibliography}

\end{document}